\documentclass[11pt]{article}
\usepackage[T1]{fontenc}
\usepackage{amsmath, amsthm}
\usepackage{amsfonts}
\usepackage{amssymb}
\usepackage{graphicx}
\usepackage{amsthm}
\usepackage{listings}
\usepackage{tikz}
\usepackage{booktabs}
\usetikzlibrary{trees}
\usepackage{float}
\usepackage{xcolor}

\definecolor{codegreen}{rgb}{0,0.6,0}
\definecolor{codegray}{rgb}{0.5,0.5,0.5}
\definecolor{codepurple}{rgb}{0.58,0,0.82}
\definecolor{backcolour}{rgb}{0.9,0.9,0.9}

\lstdefinestyle{mystyle}{
	backgroundcolor=\color{backcolour},   
	commentstyle=\color{codegreen},
	keywordstyle=\color{magenta},
	numberstyle=\tiny\color{codegray},
	stringstyle=\color{codepurple},
	basicstyle=\ttfamily\footnotesize,
	breakatwhitespace=false,         
	breaklines=true,                 
	captionpos=b,                    
	keepspaces=true,                 
	numbersep=5pt,                  
	showspaces=false,                
	showstringspaces=false,
	showtabs=false,                  
	tabsize=2, 
	alsoletter ={_, .},
	deletekeywords={_, .},
	columns = fullflexible,
	texcl=true,
	otherkeywords={!,!=,~,*,\&,\%/\%,\%*\%,\%\%,<-,<<-},
}

\usepackage{pgfplotstable}

\usepackage{breakcites}
\RequirePackage[colorlinks,citecolor=blue,urlcolor=blue]{hyperref}
\RequirePackage{graphicx}
\usepackage[ruled,vlined]{algorithm2e}

\newlength{\mytextwidth}
\usepackage{multicol}
\usepackage{tikz, pgfplots}
\usepackage{authblk}
\pgfplotsset{compat=1.17}
\usetikzlibrary[arrows.meta,bending, calc, positioning, patterns]

\theoremstyle{plain}

\newtheorem{theorem}{Theorem}[section]
\newtheorem{lemma}[theorem]{Lemma}
\newtheorem{corollary}[theorem]{Corollary}

\newtheorem{remark}[theorem]{Remark}
\newtheorem{example}[theorem]{Example}
\newtheorem{definition}[theorem]{Definition}

\newcommand{\diff}{\, \mathrm{d}}
\newcommand{\e}{\mathrm{e}}

\usepackage[margin=1in]{geometry}

\usepackage{todonotes}

\begin{document}

\title{Efficient evaluation of risk allocations}

\author[1]{Christopher Blier-Wong\thanks{Corresponding author, \href{mailto:christopher.blierwong@utoronto.ca}{christopher.blierwong@utoronto.ca}}}
\author[2]{Hélène Cossette}
\author[2]{Etienne Marceau}

\affil[1]{Department of Statistical Sciences, University of Toronto, Canada}
\affil[2]{École d'actuariat, Université Laval, Canada}

\date{\today}

\maketitle

\begin{abstract}
Expectations of marginals conditional on the total risk of a portfolio are crucial in risk-sharing and allocation. However, computing these conditional expectations may be challenging, especially in critical cases where the marginal risks have compound distributions or when the risks are dependent. We introduce a generating function method to compute these conditional expectations. We provide efficient algorithms to compute the conditional expectations of marginals given the total risk for a portfolio of risks with lattice-type support. We show that the ordinary generating function of unconditional expected allocations is a function of the multivariate probability generating function of the portfolio. The generating function method allows us to develop recursive and transform-based techniques to compute the unconditional expected allocations. We illustrate our method to large-scale risk-sharing and risk allocation problems, including cases where the marginal risks have compound distributions, where the portfolio is composed of dependent risks, and where the risks have heavy tails, leading in some cases to computational gains of several orders of magnitude. Our approach is useful for risk-sharing in peer-to-peer insurance and risk allocation based on Euler's rule.
\end{abstract}


\section{Introduction}

Risk allocations are essential in actuarial science and quantitative risk management. Roughly speaking, risk allocation refers to redistributing a total risk to its granular risks. Risk allocation is used in various contexts, such as insurance pricing, reinsurance, and regulatory capital requirements. For instance, one may allocate the total risk to the policyholders for peer-to-peer insurance pricing. For an insurance company with many lines of business, one may allocate the total risk to each line of business to determine the required capital for each line. 

Numerical evaluation of aggregate and compound distributions is a challenging problem in non-life actuarial science and quantitative risk management. Typical approaches for the evaluation of aggregate and compound distributions are using the direct convolution approach, the Panjer recursion, or the transform-based techniques like the fast Fourier transform (FFT), see \cite{wang1996premium,embrechts2009panjer,mcneil2015quantitative} for details. The direct convolution approach is computationally expensive, while the Panjer recursion is limited to compound distributions. The transform-based techniques are more efficient than the previous methods and can be used for a wide range of distributions. Thanks to efficient algorithms like the FFT, one can compute the aggregate distribution of a large portfolio of risks in a reasonable amount of time, essential for large-scale risk aggregation problems in actuarial science and quantitative risk management. 

Numerical evaluation of risk allocations is a more challenging problem than computing the aggregate distribution since it requires computing (i) the conditional expectations of the marginals given the total risk and (ii) the probability mass function of the aggregate random variable. While there exists a growing literature on risk-sharing and risk allocation, along with the properties of risk-sharing rules, the literature on the numerical evaluation of risk allocations is scarce, and the methods are often limited to small pools or a few lines of business. Solving this problem for a large portfolio of risks is especially tedious, especially when the marginal risks have compound distributions or when the risks are dependent. Providing efficient algorithms to compute the conditional expectations of marginals given the total risk is essential for risk-sharing and risk allocation problems in practical situations, and we aim in this paper to address this limitation. 

In this paper, we introduce a transform-based method to compute the conditional expectations of the marginals given the total risk, that is, $E[X_i \vert S = s]$, where $s$ lies on a lattice-type support. We show that the ordinary generating function of unconditional expected allocations is a function of the multivariate probability generating function of the portfolio. The generating function method allows us to develop transform-based techniques to compute unconditional expected allocations. 
We provide a recursive method to compute unconditional expected allocations when marginals are compound distributions. Our findings are helpful for risk-sharing and risk allocation problems, allowing computational gains that make analysis feasible for large portfolios of risks.

The main result of this paper is a representation of the ordinary generating function (OGF) for unconditional expected allocations in terms of the multivariate probability generating function (pgf) of the random vector $\boldsymbol{X}$. 
This representation enables new techniques to compute the unconditional expected allocation $E[X_i \times 1_{\{S = s\}}]$, for $i \in \{1, \dots, n\}$ and $s \in h\mathbb{N}$ efficiently. 
In some cases, we can invert the OGF of unconditional expected allocations analytically to obtain formulas expressed as partial sums depending on the pmf of $S$. In other cases, we can use an efficient algorithm, such as the FFT, to compute the unconditional expected allocations directly from the OGF; this algorithm leads to significant computational gains that allow us to compute the unconditional expected allocations for large portfolios of risks with an accuracy unachievable with existing methods.

Our approach provides a generating function method instead of a direct computation method. The method of generating functions is standard in discrete mathematics and number theory. One can use generating functions to compute a term in the Fibonacci sequence, find the sequence average, study recurrence relations or prove combinatorial identities; see \cite{wilf2006generatingfunctionology} for details. Generating functions often provide convenient and elegant methods to compute or extract terms from a sequence where closed-form formulas would be tedious. A problem encountered within this paper is related to finding the number of partitions (restricted partitions in some cases, see Section \ref{sec:bern} for an example with Bernoulli rvs) to determine if two agents can share a risk for a given outcome of the total risk. It isn't surprising to observe that generating functions can solve these problems within the context of discrete rvs: these methods are also used to compute the pmf of aggregate rvs or compound distributions; see, for example, \cite{embrechts1993applications}, \cite{grubel1999computation}, and \cite{embrechts2009panjer} for applications of generating functions for risk management. The latter authors use the expression transform approaches to refer to generating function methods, and we will follow their nomenclature in what follows. 

The paper is organized as follows. Section \ref{sec:background} provides background on risk allocations and the problem of computing conditional expectations of marginals given the total risk, along with notation and a review of existing approaches. Section \ref{sec:agf} introduces the ordinary generating function for unconditional expected allocations and reveals a relationship with the multivariate probability generating function. We also propose an efficient method to extract the unconditional expected allocations from the ordinary generating function. Section \ref{sec:ab0} presents expressions for the unconditional expected allocation in the case of (compound) Katz distributions. The results in Section \ref{sec:ab0} provide recursive formulas for unconditional expected allocations based on the pmf of the aggregate random variable, but we prove these results by using a transform-based approach and returning to the original space to compute conditional expectations. It is not always possible to analytically return to the original space; in this case, we explore applications of the FFT algorithm to compute the unconditional expected allocations in Section \ref{sec:fft}. In Section \ref{sec:dependent}, we illustrate our method to dependent risks, sometimes obtaining recursive relations, other times using the FFT algorithm. Section \ref{sec:discussion} discusses further generalizations of our results for continuous random variables.

\section{Numerical evaluation of risk allocations}\label{sec:background}

\subsection{Risk-sharing and risk allocation}

Pooling or aggregating risks is a core principle in insurance and risk management, whether in traditional insurance settings or peer-to-peer insurance models. In both cases, the aim is to reduce the variability of the total risk by pooling the risks of different policyholders or lines of business. This aggregation is beneficial as long as the risks are not perfectly correlated or comonotonic due to the diversification effect. Understanding the total aggregated risk distribution is crucial for computing essential risk measures like the Value-at-Risk (VaR) or the Tail-Value-at-Risk (TVaR, also known as expected shortfall or conditional tail expectation in the finance and banking literature). These measures help risk managers understand the portfolio's potential losses and set aside capital to cover these losses in compliance with regulatory requirements. 

In peer-to-peer insurance, the concept of risk aggregation is applied by having policyholders pool their risks. Each member shares in the total risk and pays a premium that reflects their contribution to the pooled risk. These premiums are determined post-event, based on the actual losses experienced by the group and according to a predetermined risk-sharing rule. The conditional mean risk-sharing rule, introduced by \cite{denuit2012convex}, is a popular choice in peer-to-peer insurance pricing. This rule states that the price for the $i$th participant is the expected contribution of risk $X_i$, given that the actual loss $S$ is $s$. The authors of \cite{denuit2021risksharing} list twelve desirable properties for risk-sharing rules and prove that the conditional mean risk-sharing rule satisfies eleven of them. In \cite{jiao2022axiomatic}, the authors provide an axiomatic approach to risk-sharing rules and show that the conditional mean risk-sharing rule satisfies the axioms of actuarial fairness, risk fairness, risk anonymity, and operational anonymity.

\subsection{Notation}

We now set the notation used throughout the paper. Let $\mathbb{N}$ be the set of non-negative integers $\{0, 1, 2, \dots\}$, while $\mathbb{N}_1$ be $\mathbb{N} \setminus \{0\}$. Consider a portfolio of $n$ risks $\boldsymbol{X} = (X_1, \dots, X_n)$ where each random variable (rv) has a distribution supported on a lattice-type set $h \mathbb{N} = \{hk | k \in \mathbb{N}\}$, for some fixed $h \in \mathbb{R}^+$. The multivariate cumulative distribution function (cdf) and probability mass function (pmf) are respectively $F_{\boldsymbol{X}}(\boldsymbol{x})$ and $f_{\boldsymbol{X}}(\boldsymbol{x})$, for $\boldsymbol{x} \in \{h\mathbb{N}\} \times \dots \times \{h\mathbb{N}\} = \{h\mathbb{N}\}^n$ and marginal cdfs and pmfs are noted $F_{X_i}(x) = \Pr(X_i \leq x)$ and $f_{X_i}(x) = \Pr(X_i = x)$, for $x \in h\mathbb{N}$ and $i \in \{1, \dots, n\}$. Throughout the paper, we assume that $E[X_i] < \infty$, for $i \in \{1,\ldots,n\}$. The rv representing the portfolio aggregate loss is $S = X_1 + \dots + X_n,$ with cdf $F_S$, pmf $f_S$ and expectation $E[S] = \sum_{i=1}^n E[X_i] < \infty$.

\subsection{Unconditional expected allocations}

In this paper, we propose methods based on ordinary generating functions to compute the unconditional expected allocation, which we define as follows. 
\begin{definition}[Unconditional expected allocation]
\label{def:ExpectedAllocation}
	Let $\boldsymbol{X}$ be a vector of rvs, each with distributions supported on a lattice-type set $h\mathbb{N}$, and $S = X_1 + \dots + X_n$. 
	The unconditional expected allocation of $X_i$ to a total outcome $S=s$ is defined as $E[X_i \times 1_{\{S = s\}}],$ for $i \in \{1, \dots, n\}$ and $s \in h \mathbb{N}$.
\end{definition}

While most loss distributions in actuarial science and quantitative risk management have continuous support, the assumption of lattice-type support is not too prohibitive. It suffices to discretize the continuous support into intervals of size $h$. The value of $h$ can be set to the smallest denomination of a currency system, e.g. $h = 0.01$ for dollars. In that case, the risk allocations are accurate to the nearest penny.

In the spirit of Definition \ref{def:ExpectedAllocation}, we introduce the unconditional expected cumulative allocation defined by
\begin{equation} \label{eq:expected-cumul-allocation}
    E\left[X_i \times 1_{\{S \leq s\}}\right] = \sum_{x \in \{0,h,\ldots,s\}} E\left[X_i \times 1_{\{S = x\}}\right],
\end{equation}
for $s \in h\mathbb{N}$ and $i \in \{1,\ldots,n\}$. 

Although rarely considered directly, unconditional expected allocations are essential in peer-to-peer insurance pricing and risk allocation based on Euler's rule. 

One is interested in computing a participant's contribution according to a risk-sharing rule in peer-to-peer insurance pricing schemes. The conditional mean risk-sharing rule, studied in, for instance, \cite{denuit2012convex}, is a popular choice, where the price for the $i$th participant is the expected contribution of risk $X_i$, for $i \in \{1, \dots, n\}$, given that the actual loss $S$ is $s$, that is,
\begin{equation}\label{eq:conditional-mean-risk-sharing}
	E[X_i \vert S = s] = \frac{E[X_i \times 1_{\{S = s\}}]}{\Pr(S = s)}, 
\end{equation}
assuming $\Pr(S=s) \neq 0$, for $s \in h\mathbb{N}$.
We note that $\sum_{i=1}^n E[X_i \vert S = s] = s$, meaning that, under the conditional mean sharing rule, the total contributions of all participants equal the total losses $s$, for $s \in h\mathbb{N}$. 

Unconditional expected allocations also appear in risk allocation based on Euler's rule. Regulatory capital requirements are risk measures based on the aggregate rv of an insurance company's portfolio. One risk measure of theoretical and practical interest is the TVaR. Risk allocation is an important research area in actuarial science, quantitative risk management and operations research, which aims to compute the contribution of each risk based on the total required (or available) capital. When using the TVaR as a regulatory capital requirement risk measure, one may compute the contributions of each risk to the capital based on the Euler risk-sharing paradigm; see \cite{mcneil2015quantitative} for details.

Following \cite{embrechts2013note}, define the generalized inverse of $S$ at level $\kappa$ by
$$F^{-1}_S(\kappa) = \inf_{x \in \mathbb{R}} \left\{F_{S}(x) \geq \kappa\right\},$$
for $0 < \kappa < 1$. 
We will consider the Range-Value-at-Risk as a risk measure for capital requirements, which can be seen as a generalization of the VaR and the TVaR. 
\begin{definition}
	The Range-Value-at-Risk is defined as 
	$$\mathrm{RVaR}_{\alpha_1, \alpha_2}(X) = \begin{cases}
		\frac{1}{\alpha_2 - \alpha_1} \int_{\alpha_1}^{\alpha_2} \mathrm{VaR}_{u}(X) \diff u, & \alpha_1 < \alpha_2\\
		\mathrm{VaR}_{\alpha_1}(X), & \alpha_1 = \alpha_2
	\end{cases},$$
	for $0 \leq \alpha_1 \leq \alpha_2 \leq 1$. 
\end{definition}
Clearly, the Range-Value-at-Risk becomes the Value-at-Risk $\mathrm{VaR}_\kappa(X)$ if $\alpha_1 = \alpha_2 = \kappa$. Further, if $\alpha_1 = \kappa$ and $\alpha_2 = 1$, then the corresponding RVaR is also $\mathrm{TVaR}_\kappa(X)$, also called the expected shortfall. We have
\begin{align*}
    \mathrm{RVaR}_{\alpha_1, \alpha_2}(S) &= \frac{1}{\alpha_2 - \alpha_1} \left\{F_S^{-1}(\alpha_1)[F_S(F_S^{-1}(\alpha_1)) - \alpha_1] + E\left[S \times 1_{\{F_S^{-1}(\alpha_1) <S \leq F_S^{-1}(\alpha_2)\}} \right]\right.\\
    &\qquad \qquad \qquad + \left. \vphantom{E\left[S \times 1_{\{F_S^{-1}(\alpha_1) <S \leq F_S^{-1}(\alpha_2)\}} \right]} F_S^{-1}(\alpha_2)[\alpha_2 - F_S(F_S^{-1}(\alpha_2))]\right\}.
\end{align*}

Applying Euler's rule for risk allocation \cite{tasche1999risk}, the contribution of $X_i$ to the RVaR of the aggregate random variable is $\mathrm{RVaR}_{\alpha_1, \alpha_2}(X_i; S) = E[X_i | S = \mathrm{VaR}_{\alpha_1}(S)]$ for $\alpha_1 = \alpha_2$ and
\begin{align}
	\mathrm{RVaR}_{\alpha_1, \alpha_2}(X_i; S) & = \frac{1}{\alpha_2 - \alpha_1} \left(E\left[X_i \times 1_{\left\{S = F_{S}^{-1}(\alpha_1)\right\}}\right] \frac{F_S\left(F_{S}^{-1}(\alpha_1)\right) - \alpha_1}{\Pr\left(S = F_{S}^{-1}(\alpha_1)\right)} \right. \nonumber\\
	& \qquad\qquad\qquad +
	E\left[X_i \times 1_{\left\{F_{S}^{-1}(\alpha_1) < S \leq F_{S}^{-1}(\alpha_2)\right\}}\right]\label{eq:rvar-decomp}\\
	& \qquad\qquad\qquad +\left.
	E\left[X_i \times 1_{\left\{S = F_{S}^{-1}(\alpha_2)\right\}}\right] \frac{\alpha_2 - F_S\left(F_{S}^{-1}(\alpha_2)\right)}{\Pr\left(S = F_{S}^{-1}(\alpha_2)\right)}
	\right),\nonumber
\end{align}
for $\alpha_1 < \alpha_2$. Notice that two of the three expected values in \eqref{eq:rvar-decomp} are unconditional expected allocations, and we can compute the third one using unconditional expected allocations. 
 
The Euler-based RVaR decomposition is a top-down risk allocation method of risk allocation. By the additive property of the expected value, the full allocation property \cite{mcneil2015quantitative} holds:
\begin{equation}\label{eq:full-allocation}
	\mathrm{RVaR}_{\alpha_1, \alpha_2}(S) = \sum_{i = 1}^{n}\mathrm{RVaR}_{\alpha_1, \alpha_2}(X_i; S)
\end{equation}
for any pair $(\alpha_1, \alpha_2)$ such that $0 \leq \alpha_1 \leq \alpha_2 \leq 1.$

The relations in \eqref{eq:conditional-mean-risk-sharing} and \eqref{eq:rvar-decomp} require the computation of the unconditional expected allocation $E\left[X_i \times 1_{\{S = s\}}\right]$ for $s \in h\mathbb{N}$. One, therefore, seeks an efficient method to compute these values.

\subsection{Existing approaches to compute unconditional expected allocations}

One finds two common approaches to computing unconditional expected allocations in the actuarial science and quantitative risk management literature. The first approach is a direct method for computing unconditional expected allocations through summation or integration. Letting $S_{-i} = \sum_{j = 1, j \neq i}^n X_j$, we have
\begin{equation*}
    E\left[X_1 \times 1_{\{S = s\}}\right] = \sum_{x \in \{0, h, 2h, \dots, s\}} x f_{X_1, S_{-1}}(x, s-x),
    \quad s\in h\mathbb{N}.
\end{equation*}
The direct summation method is used in \cite{cossette2018dependent} for discrete rvs when the dependence structure is an Archimedean copula. In Section 5 of \cite{barges2009tvarbased}, the authors use this approach in a continuous setting to compute TVaR-based allocations for a mixture of exponential distributions linked through an FGM copula. The second approach is based on size-biased transforms, used notably in \cite{furman2005risk, furman2008economic}; see also \cite{arratia2019size} for a review of the size-biased transform and its applications. Under that approach,
\begin{equation}\label{eq:allocation-size-biased}
	E\left[X_1 \times 1_{\{S = s\}}\right] = E[X_1] \Pr(\widetilde{X}_1 + S_{-1} = s),
	\quad s\in h\mathbb{N},
\end{equation}
where $\widetilde{X}_1$ is the size-biased transform of $X_1$ with pmf
\begin{equation}\label{eq:pmf-size-biased}
	f_{\widetilde{X}_1}(x) = x f_{X_1}(x)/E[X_1], \quad x \in h\mathbb{N}.
\end{equation}
The authors of, for instance, \cite{denuit2012convex, denuit2020sizebiased, denuit2020largeloss}, use the size-biased transform method to derive properties and results about the conditional mean risk-sharing rule. 

To simplify the notation in the theory developed in the remainder of this paper, we set $h = 1$; that is, we consider only rvs which have integer support. Therefore, for the remainder of this paper, the $n$-variate random vector $\boldsymbol{X}$ takes values in $\mathbb{N}^n$. One may transform a rv with lattice-type support $h\mathbb{N}$ into one of integer support $\mathbb{N}$ by multiplying the rv by the constant $h^{-1}$. By linearity of the unconditional expected allocation, one may easily recover unconditional expected allocations for the original rv. 

Typically, conditional mean risk-sharing and risk allocation are used for small pools or a few lines of businesses. However, the efficient methods based on OGFs proposed in this paper enable these techniques to be used even for a large portfolio of risks. Once equipped with generating functions for unconditional expected allocations, risk managers can perform risk allocation at the customer level. 

\section{Ordinary generating functions for unconditional expected allocations}\label{sec:agf}

\subsection{Ordinary generating functions}

Ordinary generating functions are a useful mathematical tool since they capture every sequence value into one formula. See Chapter 7 of \cite{graham1994concrete} or the monograph \cite{wilf2006generatingfunctionology} for details on generating functions, and Chapter 3 of \cite{sedgewick2013introduction} for efficient algorithms to extract the values of the sequence. Following Section 3.1 of \cite{sedgewick2013introduction}, we define ordinary generating functions.
\begin{definition}[Ordinary generating function]
	For a sequence $\{a_k\}_{k \in \mathbb{N}}$, the function
	\begin{equation} \label{eq:OGF}
	  A(z) = \sum_{k = 0}^\infty a_k z^k, \quad |z| \leq 1,  
	\end{equation}
	is its ordinary generating function (OGF). We use the notation $[z^k]A(z)$	to refer to the coefficient $a_k$, $k \in \mathbb{N}$.
\end{definition}

The following lemma summarizes the relevant operations one can perform on generating functions (see Theorem 3.1 and Table 3.2 of \cite{sedgewick2013introduction} for details).
\begin{lemma} \label{th:OperationsOGF}
If $A(z) = \sum_{k = 0}^\infty a_k z^k$ and $B(z) = \sum_{k = 0}^\infty b_k z^k$ are two OGFs, then the following operations produce OFGs with the corresponding sequences:
\begin{enumerate}
	\item\label{prop:addition} Addition $A(z) + B(z) = \sum_{k = 1}^{\infty} (a_{k} + b_{k})z^k$.
	\item\label{prop:rightshift} Right shift $zA(z) = \sum_{k = 1}^{\infty} a_{k-1}z^k$.
	\item\label{prop:index-multiply} Index multiply $A'(z) = \sum_{k = 0}^{\infty} (k+1)a_{k+1}z^k$.
	\item\label{prop:convolution} Convolution $A(z)B(z) = \sum_{k = 0}^{\infty} \left(\sum_{j = 0}^{k} a_j b_{k-j}\right)z^k$.
	\item\label{prop:partial-sum} Partial sum $A(z)/(1-z) = \sum_{k = 0}^{\infty} \left(\sum_{j = 0}^{k} a_j\right)z^k$.
\end{enumerate}
\end{lemma}

When $a_k \geq 0$ for $k \in \mathbb{N}$ and $\sum_{k=0}^\infty a_k= 1$, then the OGF is the pgf of a discrete rv $X$, denoted as $\mathcal{P}_{X}$, where the values of the pmf of the rv $X$ is $f_X(k) = \Pr(X=k) = a_k$, $k \in \mathbb{N}$. 
The expression in (\ref{eq:OGF}) becomes 
\begin{equation} \label{eq:PGF}
    \mathcal{P}_{X}(z) = \sum_{k=0}^\infty f_X( k) z^k, \quad |z| \leq 1.
\end{equation}
The pgf is an essential tool in all areas of probability, statistics, and actuarial science, as explained, for example, in Section 5.1 in \cite{grimmett2020probability}, and Sections 1.2 and 2.4 in \cite{panjer1992insurance}.

In this paper, we rely on a multivariate ordinary generating function capturing the values of the pmf of a discrete random vector. As described in Section 34.2.1 of \cite{johnson1997discrete}, the multivariate pgf of a vector of discrete rvs $\boldsymbol{X} = (X_1, \dots, X_n)$, with multivariate pmf $f_{\boldsymbol{X}}$, is 
\begin{equation} \label{eq:DefFgpMulti}
\mathcal{P}_{\boldsymbol{X}}(z_1, \dots, z_n)  
= E\left[z_1^{X_1} \times \cdots \times z_n^{X_n}\right] 
= \sum_{k_1 = 0}^{\infty} \cdots \sum_{k_n 
= 0}^{\infty} z_1^{k_1} \times \cdots \times z_n^{k_n} f_{\boldsymbol{X}}(k_1, \dots, k_n), 
\end{equation}
for $|z_j| \leq 1$, $j \in \{1, \dots, n\}$. For more details on multivariate pgfs or ordinary functions and their properties, see Appendix A in \cite{axelrod2015branching}, and Chapter 3 of \cite{flajolet2009analytic}.

The following theorem, a generalization of Theorem 1 of \cite{wang1998aggregation}, shows the usefulness of the multivariate pgf to capture at once the dependence relations between the components of $\boldsymbol{X}$ and aggregation of subsets of them. 

\begin{theorem}\label{thm:pgf-subset}
    Let $\mathcal{A} = \{A_1, \dots, A_m\}$ be a partition of the set $\{1, \dots, n\}$, for $m \leq n$. Define the random vector $\boldsymbol{Y}_{\mathcal{A}} = (Y_{A_1}, \dots, Y_{A_m})$ with $Y_{A_j} = \sum_{i \in A_j} X_i$, for $j \in \{1, \dots, m\}$. Then the multivariate pgf of $\boldsymbol{Y}_\mathcal{A}$ is
    \begin{equation}\label{eq:pgf-subset}
        \mathcal{P}_{\boldsymbol{Y}_{\mathcal{A}}}(t_1, \dots, t_m) = \mathcal{P}_{\boldsymbol{X}}\left(\prod_{j =1}^m t_j^{1_{\{1 \in A_j\}}}, \dots, \prod_{j =1}^m t_j^{1_{\{n \in A_j\}}}\right), \quad |t_j| \leq 1, 
        \hspace{0.5em}
        j \in \{1, \dots, m\}.
    \end{equation}
\end{theorem}
If the margins of $\boldsymbol{X}$ correspond to risks related to individual business units, then one may apply Theorem \ref{thm:pgf-subset} to obtain the pgf of random vectors at the department or division level within the organizational hierarchy of the entire business. If the margins of $\boldsymbol{X}$ correspond to individual risks in an insurance portfolio, then Theorem \ref{thm:pgf-subset} provides an expression for the pgf of total risks aggregated by coverage type or geographic regions. 

Here are special cases of partitions useful in the context of Theorem \ref{thm:pgf-subset}. When $m = 1$, such that $A_1 = \{1, \dots, n\}$ and $Y_{A_1} = S$, the result in (\ref{eq:pgf-subset}) leads to Theorem 1 of \cite{wang1998aggregation}. For $m = 2$, we have $A_1 \subset \{1, \dots, n\}$ and $A_2 = A_1^C = \{1, \dots, n\} \setminus A_1$, that is, a subset of $\{1, \dots, n\}$ and its complement. Finally, if $m = n$, then $A_j = \{j\}$ and $Y_{A_j} = X_j$, for $j \in \{1, \dots, n\}$. Remark that for each product in the arguments of \eqref{eq:pgf-subset}, only one value of $t_j$, for $j \in \{1, \dots, m\}$, remains since $\mathcal{A}$ is a partition of a set. 

As noted in Section 4.2 of \cite{wang1998aggregation} and Section 5.1 of \cite{grimmett2020probability}, one may use pgfs to extract factorial moments, mixed moments and pmfs. In the remainder of this section, we add unconditional expected allocations to this list.

\subsection{Ordinary generating functions for unconditional expected allocations}

In this paper, our interest is that of computing unconditional expected allocations; hence, we define the function $\mathcal{P}_S^{[i]}(t)$ as the OGF of the sequence of unconditional expected allocations for the rv $X_i$, that is, 
\begin{equation} \label{eq:ogfExpectedAllocation}
  \mathcal{P}_S^{[i]}(t) := \sum_{k = 0}^{\infty} t^{k}E\left[X_i \times 1_{\{S = k\}}\right],  
\end{equation}
for $i \in \{1, \dots, n\}$. 

Aiming to simplify the presentation, unless otherwise specified, we develop formulas for $i = 1$ for the remainder of this paper. One may obtain the other unconditional expected allocations by appropriate reindexing. 

The following theorem is at the basis of the results in this paper and provides a link between the OGF of the unconditional expected allocations of the rv $X_1$ and the multivariate pgf of $\boldsymbol{X}$. 

\begin{theorem} \label{th:TheBigResult}
    If $\boldsymbol{X}$ is a vector of rvs with multivariate pgf $\mathcal{P}_{\boldsymbol{X}}$ and $S$ is the aggregate loss rv, then the expression of $\mathcal{P}_S^{[1]}$ is given by
    \begin{equation} \label{eq:TheBigResult}
        \mathcal{P}_{S}^{[1]}(t) = \left[t_1 \times \frac{\partial }{\partial t_1} \mathcal{P}_{\boldsymbol{X}}(t_1, \dots, t_n)\right]_{t_1 = \dots = t_n = t}.  
    \end{equation}
\end{theorem}
\begin{proof}
    Applying Theorem \ref{thm:pgf-subset} with $m = 2$, $A_1 = 1$ and $A_2 = \{2, \dots, n\}$, the pgf of $(X_1, S_{-1})$ is
	$$\mathcal{P}_{X_1, S_{-1}}(t_1, t_{-1})  = E\left[t_1^{X_1} t_{-1}^{X_2 + \dots + X_n}\right] = \mathcal{P}_{\boldsymbol{X}}(t_1, t_{-1}, \dots, t_{-1})$$
	for $|t_1| \leq 1$ and $|t_{-1}| \leq 1$. We define
	\begin{equation*}\label{eq:eagf1}
		\mathcal{P}^{[1]}_{X_1, S_{-1}}(t_1, t_{-1})  = \sum_{k_1 = 0}^{\infty} \sum_{k_2 = 0}^{\infty} k_1t_1^{k_1} t_{-1}^{k_2} f_{X_1, S_{-1}}(k_1, k_2),
	\end{equation*}
	which becomes
	\begin{equation*}\label{eq:eagf1b}
	   \mathcal{P}^{[1]}_{X_1, S_{-1}}(t_1, t_{-1})= t_1 \frac{\partial}{\partial t_1}\sum_{k_1 = 0}^{\infty} \sum_{k_2 = 0}^{\infty} t_1^{k_1} t_{-1}^{k_2} f_{X_1, S_{-1}}(k_1, k_2) =  t_1 \times \frac{\partial }{\partial t_1} \mathcal{P}_{X_1, S_{-1}}(t_1, t_{-1}).
	\end{equation*}

	Finally, it follows from the same arguments as in Theorem \ref{thm:pgf-subset} that $\mathcal{P}^{[1]}_{S}(t) = \mathcal{P}^{[1]}_{X_1, S_{-1}}(t, t)$, which becomes
	\begin{align}
		\mathcal{P}^{[1]}_{S}(t) = \sum_{k_1 = 0}^{\infty} \sum_{k_2 = 0}^{\infty} k_1 t^{k_1} t^{k_2} f_{X_1, S_{-1}}(k_1, k_2)
		&= \sum_{k = 0}^{\infty} t^{k}\sum_{k_1 = 0}^{k} k_1 f_{X_1, S_{-1}}(k_1, k-k_1)= \sum_{k = 0}^{\infty} t^{k}E\left[X_1 \times 1_{\{S = k\}}\right],\label{eq:conditional-mean-generating-function}
	\end{align}
	where \eqref{eq:conditional-mean-generating-function} is the power series representation in \eqref{eq:ogfExpectedAllocation} of unconditional expected allocations, as desired.
\end{proof}

From the uniqueness theorem of pgfs (see, for instance, Section 5.1 of \cite{grimmett2020probability}), one can recover the values of $E\left[X_1 \times 1_{\{S = k\}}\right]$, $k \in \mathbb{N}$, by differentiating
$$[t^k]\mathcal{P}^{[1]}_{S}(t) = E\left[X_1 \times 1_{\{S = k\}}\right] = \frac{1}{k!} \left.\frac{\diff^k}{\diff t^k} \mathcal{P}^{[1]}_{S}(t)\right\vert_{t = 0},$$
or using an algorithm to extract the coefficients of a polynomial. The entire Section \ref{ss:fft} provides a method using FFT to extract the coefficients from the OGF for unconditional expected allocations. Consequently, \eqref{eq:conditional-mean-generating-function} is a powerful tool to capture every unconditional expected allocation for $X_1$ within a single function. 

An especially convenient corollary holds for allocating a rv independent from the remaining risks.
\begin{corollary}\label{cor:indep-s-x}
	If $X_1$ and $S_{-1}$ are independent, 
	then the expression of $\mathcal{P}_{S}^{[1]}$ in (\ref{eq:TheBigResult}) becomes 
	\begin{equation}\label{eq:agf-indep}
		\mathcal{P}_{S}^{[1]}(t) = t \mathcal{P}'_{X_1}(t) \mathcal{P}_{S_{-1}}(t).
	\end{equation}
\end{corollary}

Aiming for a more efficient method to compute the expectations, we find an OGF for the unconditional expected cumulative allocation defined in (\ref{eq:expected-cumul-allocation}).

\begin{corollary} \label{thm:cumul-exp-allocations}
	If $\boldsymbol{X}$ is a vector of rvs with multivariate pgf $\mathcal{P}_{\boldsymbol{X}}$ and $S$ is the aggregate rv,
	then the function $$\mathcal{P}_{S}^{[1]}(t)/(1-t) = \frac{1}{1 - t}\left[t_1 \times \frac{\partial }{\partial t_1} \mathcal{P}_{\boldsymbol{X}}(t_1, \dots, t_n)\right]_{t_1 = \dots = t_n = t} $$ 
	is the OGF of the sequence of cumulative unconditional expected allocations $\{E[X_1 \times 1_{\{S \leq k\}}]\}_{k \in \mathbb{N}}$.
\end{corollary}

\begin{proof}
	Applying operation \ref{prop:partial-sum} of Lemma \ref{th:OperationsOGF}, 
	we have
	\begin{equation}\label{eq:cumul-agf}
		\frac{\mathcal{P}^{[1]}_{S}(t)}{1 - t} = \sum_{k = 0}^{\infty} t^{k}\left(\sum_{j = 0}^{k}E\left[X_1 \times 1_{\{S = j\}}\right]\right) = \sum_{k = 0}^{\infty} t^{k}E\left[X_1 \times 1_{\{S \leq k\}}\right].
	\end{equation}
\end{proof}

\subsection{Outline of the FFT approach to compute unconditional expected allocations}\label{ss:fft}

Equipped with an OGF for unconditional expected allocations, one may seek to solve for the unconditional expected allocations analytically. 
In some cases, analytical inversion of the OGF will be possible, but otherwise, we may resort to numerical algorithms to compute the unconditional expected allocations. 
This section provides an algorithm to recover the unconditional expected allocations based on their OGF. 

A significant advantage of working with pgfs (and more generally, with OGFs) is that the FFT algorithm of \cite{cooley1965algorithm} provides an efficient method to extract the values of OGFs, as explained in Chapter 30 of \cite{cormen2009introductiona}. See also \cite{embrechts1993applications} for applications of the FFT algorithm in actuarial science and quantitative risk management.

Define the characteristic function of $S$ as 
$$\phi_{S}(t) := E\left[\e^{itS}\right] = \mathcal{P}_{S}\left(\e^{it}\right), \quad |t| \leq 1$$
and analogously, the characteristic version of the OGF for unconditional expected allocations,
$$\phi^{[1]}_{S}(t) := \sum_{k = 0}^{\infty} \e^{itk}E\left[X_1 \times 1_{\{S = k\}}\right] = \mathcal{P}^{[1]}_{S}\left(\e^{it}\right), \quad |t| \leq 1.$$
In this section, we aim to recover the values of $E[X_1 \times 1_{\{S = k\}}]$ using the discrete Fourier transform (DFT). Set $\boldsymbol{f}_X = (f_{X}(0), f_{X}(1), \dots, f_{X}(k_{max} -1))$ for a truncation point $k_{max} \in \mathbb{N}$. Here we assume that $f_X(k)= 0$ for $k \geq k_{max}$ such that there is no truncation error. The DFT of $\boldsymbol{f}_X$, noted $\widehat{\boldsymbol{f}}_X = (\widehat{f}_{X}(0), \widehat{f}_{X}(1), \dots, \widehat{f}_{X}(k_{max} -1))$, is  
\begin{equation}\label{eq:fft}
	\widehat{f}_{X}(k) = \sum_{j = 0}^{k_{max} - 1} f_{X}(j) \e^{i2\pi jk/k_{max}},\quad k = 0, \dots, k_{max} - 1.
\end{equation}
The inverse DFT can recover the original sequence with 
\begin{equation}\label{eq:fft-inv}
	f_{X}(k) = \frac{1}{k_{max}} \sum_{j = 0}^{k_{max} - 1}\widehat{f}_{X}(j)\e^{-i2\pi jk/k_{max}}, \quad k = 0, \dots, k_{max} - 1.
\end{equation}
The authors of \cite{embrechts2009panjer} explain how computing the pmf of a compound sum is more efficient with the FFT than using Panjer recursion or direct convolution. We now show how to apply the FFT algorithm to compute unconditional expected allocations. Let $\mu_{1:k} = E[X_1 \times 1_{\{S = k\}}]$ for $k = 0, \dots, k_{max} - 1$ and $\boldsymbol{\mu}_{1} = (\mu_{1:0}, \dots, \mu_{1:(k_{max} - 1)})$ with the obvious case $\mu_{1:0} = 0$. Then, the discrete Fourier transform of $\boldsymbol{\mu}_1$, noted $\widehat{\boldsymbol{\mu}}_{1} = (\widehat{\mu}_{1:0}, \dots, \widehat{\mu}_{1:(k_{max}-1)})$, is 
\begin{equation}\label{eq:dft-allocation}
	\widehat{\mu}_{1:j} = \mathcal{P}^{[1]}_{S}\left(\e^{i2\pi j/k_{max}}\right), \quad j = 0, \dots, k_{max}-1.
\end{equation}
For notational convenience, we write the vector $\{\e^{i2\pi j/k_{max}}\}_{0\leq j \leq k_{max - 1}}$ as $\widehat{\boldsymbol{\e}}_1$. Then, we have that $\widehat{\boldsymbol{\mu}}_{1} = \mathcal{P}^{[1]}_S(\widehat{\boldsymbol{\e}}_1)$. Computing the inverse DFT of \eqref{eq:dft-allocation} yields the values of $E[X_1 \times 1_{\{S = k\}}]$ for $k = 0, \dots, k_{max}-1$. If $k_{max}$ is a power of 2, algorithms like the FFT of \cite{cooley1965algorithm} are especially efficient. 

Note that computing the cumulative unconditional expected allocations is trickier since division by $(1-t)$ is undefined for $|t| = 1$. One, therefore, requires simplifications before applying the FFT algorithm to the OGF of cumulative unconditional expected allocations. In practice, one only obtains a slight numerical advantage from using the FFT algorithm for cumulative unconditional expected allocations, which has algorithmic complexity $O(n \log n)$. Suppose one computes unconditional expected allocations with the FFT algorithm and takes the cumulative sum of the result. In that case, the algorithmic complexity remains $O(n \log n)$.

One consideration when using the FFT algorithm is that one must select a truncation point large enough such that $f_{S}(k_{max}) = 0$. One could have a large value of $k_{max}$ if $S$ is a large portfolio or if individual risks have heavy tails. In the context of peer-to-peer insurance with a stop loss reinsurance contract with trigger $\omega$, we have $f_{S}(x) = 0$ for $x > \omega$, and $E\left[X_1 \times 1_{\{S = k\}}\right] = E\left[X_1 \times 1_{\{S = \omega\}}\right]$ for all $k \geq \omega$; thus stop loss contracts sets an upper bound to the truncation point required. 

If $X_1$ is a discrete rv, independent of $S_{-1}$, then the OGF for unconditional expected allocations is given by \eqref{eq:agf-indep}. If we have no closed-form solution for $\mathcal{P}_{X_1}'(t)$, then one can compute the DFT of $t\mathcal{P}_{X_1}'(t)$ by using the pmf of $X_1$ and the properties of OGFs. One can compute the pgf of $X_1$ as $\mathcal{P}_{X_1}(t) = \sum_{k = 0}^{\infty}t^k f_{X_1}(k)$ and $t\mathcal{P}_{X_1}'(t) = \sum_{k = 0}^{\infty} t^k kf_{X_1}(k)$. It follows that one can compute the DFT of $t\mathcal{P}_{X_1}'(t)$ as the DFT of the vector $\left\{k f_{X_1}(k)\right\}_{k \in \mathbb{N}}$. We can compute the DFT of $t\mathcal{P}_{X_1}'(t)/(1-t)$ as the DFT of the partial sum of the vector $\left\{k f_{X_1}(k)\right\}_{k \in \mathbb{N}}$.

To apply the generating function approach with the FFT algorithm (or other efficient convolution algorithms) when the rvs are continuous, one must discretize their continuous cdfs for a step size $h \in \mathbb{R}^+$.  
For a brief presentation of the upper, lower, and mean preserving discretization methods and their applications with the FFT algorithm, 
see, for instance, Section 5 of \cite{barges2009tvarbased} and Section 2 of \cite{embrechts2009panjer}. Stochastic order properties for each of these three methods are examined in Chapter 1 of \cite{muller2002comparison}.

\section{Implications for Katz distributions}\label{sec:ab0}

Let $M$ be a positive discrete rv following a distribution belonging to the Katz family of distributions (see \cite{katz1965unified}, Section 2.5.4 of \cite{winkelmann2008econometric} and Section 2.3.1 of \cite{johnson2005univariate}), also referred to as the $(a,b,0)$ family of distributions in \cite{klugman2018loss}. The pmf of $M$ satisfies the recursive relation 
\begin{equation} \label{eq:KatzPmfRecursive}
    f_M(k) = (a + b/k) f_M(k-1), \quad k\in \mathbb{N}_1,
\end{equation}
where $a<1$ and $b>0$. The expectation is $E[M] = b/(1-a)$, while the variance is $Var(M) = b/(1-a)^2$. One derives the following result for the pgf of $M$ from (\ref{eq:KatzPmfRecursive}). 
\begin{lemma}\label{lemma:pmf-deriv}
    If $M$ follows a Katz distribution, then its pgf satisfies the differential equation $\mathcal{P}_M'(t) = (a + b)/(1 - at) \mathcal{P}_M(t)$. 
\end{lemma}
\begin{proof}
    See Section 4.5.1 of \cite{dickson2017insurance}. 
\end{proof}

Note that the solution to the differential equation in Lemma \ref{lemma:pmf-deriv}, as provided in equation (2.41) of \cite{johnson2005univariate}, is $\mathcal{P}_M(t) = [(1-a)/(1-at)]^\wedge (b/a+ 1)$, for $|t| \leq 1$ and $a \neq 0$.

Members of the Katz family are the Poisson distribution (with $a = 0$ and $b = \lambda$), the binomial distribution (with $a = -q/(1-q)$ and $b = (n + 1)q/(1-q)$) and the negative binomial distribution with pmf given by 
$$f_{M}(k) = \binom{r + k - 1}{k} q^r (1-q)^k, \quad k \in \mathbb{N},$$
with $a = 1-q$ and $b = (r-1)(1-q)$. Note that each distribution has different starting values within the recursive relation (provided in the references above), but the starting values aren't required in the current paper. 

\subsection{Allocations and family of Katz distributions}

The following theorem presents an efficient formula to compute unconditional expected allocations.
\begin{theorem}\label{thm:ab0}
	Let $X_1$ follow a Katz distribution independent of $S_{-1}$. For $|a| < 1$ and $k \in \mathbb{N}_1$, we have
	\begin{equation}\label{eq:allocation-ab0}
		[t^k]\mathcal{P}^{[1]}_S(t) = E\left[X_1 \times 1_{\{S = k\}}\right] = (a+b)\sum_{j = 0}^{k-1} a^{j}f_{S}(k-1-j)
	\end{equation}
	and 
	\begin{subequations}
	\begin{align}
		[t^k]\left\{\frac{\mathcal{P}^{[1]}_S(t)}{1-t}\right\} = E\left[X_1 \times 1_{\{S \leq k\}}\right] &=(a + b) \sum_{j = 0}^{k-1}a^{j} F_S(k - 1 - j)\label{eq:cumul-allocation-ab0-1}\\
		&=(a + b) \sum_{j = 0}^{k-1} \frac{1 - a^{j+1}}{1 - a}f_{S}(k-1-j).\label{eq:cumul-allocation-ab0-2}
	\end{align}	
	\end{subequations}
\end{theorem}

\begin{proof}
	Applying Corollary \ref{cor:indep-s-x} and Lemma \ref{lemma:pmf-deriv}, the OGF for unconditional expected allocations is
	\begin{equation}\label{eq:pgfab0}
		\mathcal{P}^{[1]}_S(t) = t\mathcal{P}'_{X_1}(t) \mathcal{P}_{S_{-1}}(t) = t\frac{a + b}{1 - at} \mathcal{P}_{X_1}(t)\mathcal{P}_{S_{-1}}(t) = t\frac{a + b}{1 - at} \mathcal{P}_{S}(t).
	\end{equation}
	Then, \eqref{eq:allocation-ab0} follows from Property \ref{prop:convolution} of OGFs in Lemma \ref{th:OperationsOGF}. The relation in \eqref{eq:cumul-allocation-ab0-1} follows from another application of Property \ref{prop:convolution} of Lemma \ref{th:OperationsOGF} to \eqref{eq:allocation-ab0}. Alternatively, the OGF for cumulative unconditional expected allocations is
	\begin{equation}\label{eq:pgfab0-cumul}
		\frac{\mathcal{P}^{[1]}_S(t)}{1-t} =  (a + b)\frac{t}{(1 - at)(1-t)} \mathcal{P}_{S}(t) = \frac{a+b}{a-1}\mathcal{P}_{S}(t)\left(\frac{1}{1 - at}-\frac{1}{1-t}\right) = \frac{a+b}{a-1}\mathcal{P}_{S}(t) \sum_{k = 0}^{\infty} t^k \left(a^k - 1\right)
	\end{equation}
	for $|t|<1$. Then, \eqref{eq:cumul-allocation-ab0-2} also follows from the convolution property of OGFs in Lemma \ref{th:OperationsOGF}. 
\end{proof}

Notice that \eqref{eq:allocation-ab0} and \eqref{eq:cumul-allocation-ab0-2} require the same number of computations, so it isn't more complex to compute cumulative unconditional expected allocations than individual valued allocations. We also have the relationship
\begin{equation}\label{eq:allocation-vs-cumul-allocation}
	(a-1)E\left[X_1 \times 1_{\{S \leq k\}}\right] = aE\left[X_1 \times 1_{\{S = k\}}\right]- (a+b)F_{S}(k-1)
	,\quad |a|< 1, \quad k \in \mathbb{N}.
\end{equation}

We list the implications of Theorem \ref{thm:ab0} in Table \ref{tab:ab0-implications}, which hold whenever $X_1$ and $(X_2, \dots, X_n)$ are independent, 
even if the random vector $(X_2, \dots, X_n)$ has a complicated dependence structure. The following two examples are special cases of Theorem \ref{thm:ab0} with practical interest. Both examples show that by first writing the problem in the transformed space, we can then invert the OGF back to the original probability space to obtain closed-form or recursive-type expressions for unconditional expected allocations. In these cases, we do not require the FFT algorithm to compute the unconditional expected allocations directly (although one may need to use the FFT algorithm to compute the pmf of $S$). 

\begin{example}[Poisson distributions]
Assume that $X_1, \dots, X_n$ are independent with $X_i \sim Pois(\lambda_i)$, $i \in \{1,\ldots,n\}$. Then, $S \sim Pois(\lambda_S)$, with $\lambda_S = \lambda_1 + \dots + \lambda_n$. From (\ref{eq:allocation-ab0}) of Theorem \ref{thm:ab0}, we recover the result presented in Section 10.3 of \cite[page 413]{marceau2013modelisation}, 
$$[t^k]\mathcal{P}_S^{[1]}(t) = E\left[X_1 \times 1_{\{S = k\}}\right] = \lambda_1 \frac{\lambda_S ^{k-1}e^{-\lambda_S}}{(k-1)!} = \frac{\lambda_1}{\lambda_S} k \Pr(S = k), \quad k \in \mathbb{N}.$$
Thus, we have $E[X_1 | S = k] = \lambda_1/\lambda_S k$, which is a linear function of $k$, hence the contribution under the conditional mean risk-sharing rule coincides with the contribution under the proportional (or linear) allocation rule. 
\end{example}

\begin{example}[Negative binomial distributions]
Assume that $X_1, \dots, X_n$ are independent with $X_i \sim NB(r_i, q_i)$, $i \in \{1,\ldots,n\}$. Inserting $a = (1-q_1)$ and $b = (r_1-1)(1-q_1)$ into \eqref{eq:pgfab0}, we have
\begin{align*}
	\mathcal{P}_{S}^{[1]}(t) = t\frac{r_1(1-q_1)}{1 - (1-q_1)t} \mathcal{P}_{S}(t) 
	&= \frac{r_1(1-q_1)}{q_1} t \left(\frac{q_1}{1-(1-q_1)t}\right)^{r_1+1}\mathcal{P}_{S_{-1}}(t), \quad |t| \leq 1.
\end{align*}
We can define $$\mathcal{P}_{S^*}(t) := t \left(\frac{q_1}{1-(1-q_1)t}\right)^{r_1+1}\mathcal{P}_{S_{-1}}(t),$$
which corresponds to the pgf of a rv $S^*$ whose pmf is the convolution of the pmfs of $n$ negative binomial distributed rvs, shifted to the right by one. 
It follows that $\mathcal{P}_{S}^{[1]}(t) = r_1(1-q_1)/q_1\mathcal{P}_{S^*}(t)$, which we can invert to the original space and apply Theorem 1 of \cite{furman2007convolution} to obtain
\begin{align*}
	[t^k]\mathcal{P}_{S}^{[1]}(t) = E\left[X_1 \times 1_{\{S = k\}}\right] &= r_1 \frac{1-q_1}{q_1} R \sum_{\ell = 0}^{\infty} \delta_{\ell} \binom{r+\ell + k -1}{k} (q^*)^{r+1+\ell} (1-q^*)^{k-1},
\end{align*}
where $q^* = \min(q_1, \dots, q_n)$, $r = r_1 + \dots + r_n$,
$$R = \left(\frac{q^*}{1-q^*}\right)^n\prod_{j = 1}^{n}\frac{(1-q_j)}{q_j}$$
$$\delta_{\ell+1} = \frac{1}{\ell+1} \sum_{i = 1}^{\ell+1}i \xi_i \delta_{\ell+1-i}; \quad \delta_0 = 1; \quad \ell \in \mathbb{N}_0$$
and
$$\xi_i = \frac{r_1 + 1}{i} \left(1 - \frac{(1-q^*)q_1}{q^*(1-q_1)}\right) + \sum_{j = 2}^{n} \frac{r_j}{i} \left(1-\frac{(1-q^*)q_j}{q^*(1-q_j)}\right).$$
\end{example}

For binomial distributions, one requires the success probability to satisfy $q < 1/2$ such that $|a| < 1$. Alternately, if $q > 1/2$, one could express the problem in terms of failure probability $1 - q$ and then apply Theorem \ref{thm:ab0}. 
\begin{table}[ht]
	\centering
	\resizebox{\textwidth}{!}{
	\begin{tabular}{cccc}
		\hline
		&       Poisson        &                         Negative binomial                         &                                            Binomial                                            \\ \hline
		$a$                &          0           &                               $1-q$                               &                                           $-q/(1-q)$, for $0 < q < 1/2$                                           \\
		$b$                &      $\lambda$       &                           $(r-1)(1-q)$                            &                                         $(n+1)q/(1-q)$                                         \\
		$E[X_1 \times 1_{\{S = k\}}]$     & $\lambda f_S(k - 1)$ &               $r\sum_{j = 1}^{k}(1-q)^{j} f_S(k-j)$               &                 $ n \sum_{j = 1}^{k} (-1)^{j+1} \left(\frac{q}{1-q}\right)^{j} f_S(k-j)$                 \\
		$E[X_1 \times 1_{\{S \leq k\}}]$ (v1) & $\lambda F_S(k - 1)$ &               $r\sum_{j = 1}^{k}(1-q)^{j} F_S(k-j)$               &                 $ n \sum_{j = 1}^{k}(-1)^{j+1} \left(\frac{q}{1-q}\right)^{j} F_S(k-j)$                 \\
		$E[X_1 \times 1_{\{S \leq k\}}]$ (v2) &      $\lambda F_S(k - 1)$              & $r\frac{1-q}{q}\sum_{j = 1}^{k}\left(1 - (1-q)^j\right) f_S(k-j)$ & $ -n \sum_{j = 1}^{k}\frac{1 - \left(-\frac{q}{1-q}\right)^{j+1}}{1 + \frac{q}{1-q}} f_S(k-j)$\\
		\hline
	\end{tabular}}
	\caption{Implications of Theorem \ref{thm:ab0} for all distributions.}\label{tab:ab0-implications}
\end{table}

\subsection{Allocations and family of compound Katz distributions}

Let $M$ be a frequency rv with support on $\mathbb{N}$. Let $\{B_1, B_2, \dots\}$ form a sequence of independent, identically distributed and non-negative severity rvs, independent of $M$. Within the context of the current paper, we assume that the severity rvs take values in $\mathbb{N}$. In this section, we consider cases where the rv $X$ is defined as a random sum, that is,
\begin{equation} \label{eq:DefXSommeAleatoire}
    X = \begin{cases}
	0, & M = 0\\
	\sum_{j = 1}^{M} B_j,& M > 0
\end{cases}.    
\end{equation}
It follows from \eqref{eq:DefXSommeAleatoire} that 
the pmf of $X$ is 
$$f_X(k) = 
\begin{cases}
	\Pr(M = 0), & k = 0\\
	\sum_{j = 1}^{\infty} \Pr(M = j) \Pr(B_1 + \dots + B_j = k), & k \in \mathbb{N}
\end{cases}.$$
Evaluation of $\Pr(B_1 + \dots + B_j = k)$ is analytically and computationally expensive since direct computation results from $j-1$ convolutions. 
Fortunately, \cite{panjer1981recursivea} and others have developed efficient recursive relationships to compute the pmf of $X$ when $M$ is a Katz distribution; we often refer to these relations as Panjer recursions. We are now interested in the OGF for unconditional expected allocations for compound Katz distributions such that we may have an efficient algorithm for unconditional expected allocations.

\begin{theorem}
\label{thm:compount-ab0}
	Let $X_1$ be a rv having a compound Katz distribution with frequency rv $M_1$ having cdf in the Katz family of distributions
	with parameter $|a| < 1$ and discrete severity rv $B_1$, with $X_1$ independent of $S_{-1}$. The OGF of unconditional expected allocations is
	\begin{equation}\label{eq:pgf_compound_ab0}
		\mathcal{P}^{[1]}_S(t) = t\mathcal{P}'_{B_1}(t) \mathcal{P}'_{M_1}\left(\mathcal{P}_{B_1}(t)\right) 
		\mathcal{P}_{S_{-1}}(t).
	\end{equation}
	Further, if $|a \mathcal{P}_{B_1}(t)| < 1$ for all $|t| < 1$, then 
	\begin{equation}\label{eq:pgf_compound_ab0-v2}
		\mathcal{P}^{[1]}_S(t) = t\mathcal{P}'_{B_1}(t) \frac{a + b}{1 - a \mathcal{P}_{B_1}(t)}\mathcal{P}_S(t).
	\end{equation}
\end{theorem}
\begin{proof}
	The pgf of the compound rv $X_1$ is $\mathcal{P}_{X_1}(t) = \mathcal{P}_{M_1}\left(\mathcal{P}_{B_1}(t)\right)$, then \eqref{eq:pgf_compound_ab0} follows directly from \eqref{eq:agf-indep}. The relation in \eqref{eq:pgf_compound_ab0-v2} follows from the chain rule and Lemma \ref{lemma:pmf-deriv}.
\end{proof}

\begin{example}[Independent compound Poisson distributions]\label{ex:compound-poisson}
	Let $X_1$ be a rv whose distribution belongs to the class of compound Poisson distributions, whose severity distribution is discrete with support $\mathbb{N}$. We have
	\begin{equation}\label{eq:agf-compound-poisson}
		\mathcal{P}^{[1]}_{S}(t) = \lambda_1 t \mathcal{P}_{B_1}'(t) \mathcal{P}_{M_1}(\mathcal{P}_{B_1}(t))\mathcal{P}_{S_{-1}}(t) = \lambda_1  t \mathcal{P}_{B_1}'(t) \mathcal{P}_{S}(t).
	\end{equation}
	It follows that
	\begin{equation*}
		[t^k]\mathcal{P}^{[1]}_S(t) = E[X_1 \times 1_{\{S = k\}}] = \lambda_1 \sum_{l = 1}^{k} l f_{B_1}(l) f_{S}(k - l), \quad k \in \mathbb{N}_1
	\end{equation*}
	and
	\begin{align*}
	    [t^k]\left\{\frac{\mathcal{P}^{[1]}_S(t)}{1-t}\right\} = E[X_1 \times 1_{\{S \leq k\}}] &= \lambda_1 \sum_{l = 1}^{k} E\left[B_1 \times 1_{\{B_1 \leq l\}}\right] f_{S}(k - l), \quad k \in \mathbb{N}_1\\
	    &= \lambda_1 \sum_{l = 1}^{k} l f_{B_1}(l) F_{S}(k - l), \quad k \in \mathbb{N}_1.
	\end{align*}
\end{example}
\begin{remark}
	The results of Theorem \ref{thm:compount-ab0} are analogous to Section 4 of \cite{denuit2020largeloss} when the severity follows a discrete distribution. 
	One can recover continuous versions of the results from \cite{denuit2020largeloss} using the continuous version of the OGF for cumulative unconditional expected allocations; see Section \ref{sec:discussion} for details.
\end{remark}

\subsection{Algorithm for a sum of independent compound Poisson distributed rvs}

Consider a portfolio of $n$ independent participants, where $X_i$ is a compound Poisson distributed rv with frequency parameter $\lambda_i$ and discrete severity rv $B_i$ for $i = 1, \dots, n$. We have $\mathcal{P}_{S}(t) = \prod_{i = 1}^{n} \mathcal{P}_{M_i}\left(\mathcal{P}_{B_i}(t)\right).$ For Poisson distributions, the OGF of unconditional expected allocations for the $i$th risk, $i \in \{1, \dots, n\}$, is $\lambda_i t\mathcal{P}'_{B_i}(t)\mathcal{P}_{S}(t).$

We use the FFT algorithm to compute the unconditional expected allocations. The most computationally intensive step is using the FFT algorithm to compute the values of $f_{S}$. Fortunately, since the term $\mathcal{P}_{S}(t)$ is present for the OGF of the probability masses of $S$, of the unconditional expected allocations and of the cumulative unconditional expected allocations, one must only compute the DFT of $\mathcal{P}_{S}(t)$ once. In Algorithm \ref{algo:expected-allocations}, we present a method to compute the unconditional expected allocations efficiently using the FFT algorithm. One can change line \ref{algoline:pmf-derivee} by the cumulative sum of the vector to compute cumulative unconditional expected allocations. 

\begin{algorithm}
	\KwIn{Parameters $\lambda_i, \boldsymbol{f}_{B_i}$ for $i = 1, \dots, n$.}
	\KwOut{Unconditional expected allocations $E[X_i \vert S = k]$ for $k = 0, \dots, k_{max}-1$ and $i = 1, \dots, n$.}
	\nl \For{$i = 1, \dots, n$}{
		\nl Compute $\widehat{\boldsymbol{f}}_{X_i} = \mathcal{P}_{X_i}(\widehat{\boldsymbol{\e}}_1)$ or with \eqref{eq:fft}\;
	}
	\nl Compute the DFT of $S$ as the element-wise product $\widehat{\boldsymbol{f}}_{S} = \prod_{i = 1}^{n}\widehat{\boldsymbol{f}}_{X_i}$\;
	\nl Compute $\boldsymbol{f}_{S}$ by taking the inverse DFT of $\widehat{\boldsymbol{f}}_{S}$\;
	\nl \For{$i = 1, \dots, n$}{
		\nl Compute the DFT $\widehat{\boldsymbol{\phi}}_{B_i}$ of the vector $\{(k+1)f_{B_i}(k+1)\}_{0 \leq k \leq k_{max}-1}$\label{algoline:pmf-derivee}\;
		\nl Compute element-wise $\widehat{\boldsymbol{\mu}}_{i} = \lambda_i \widehat{\boldsymbol{\e}}_1 \times \widehat{\boldsymbol{\phi}}_{B_i}  \times \widehat{\boldsymbol{f}}_{S}$\;
		\nl Compute $\boldsymbol{\mu}_i$ as the inverse DFT of $\widehat{\boldsymbol{\mu}}_{i}$\;
		\nl Compute $\{E[X_i \vert S = k]\}_{0\leq k \leq k_{max}-1}$ by the element-wise division $\boldsymbol{\mu}_i / \boldsymbol{f}_{S}$\;
	}
	\nl Return $\{E[X_i \vert S = k]\}_{0\leq k \leq k_{max}-1}$ for $i = 1, \dots, n$.
	\caption{Conditional means for compound Poisson distributions.} \label{algo:expected-allocations}
\end{algorithm}

\section{Applications of the FFT algorithm}\label{sec:fft}

In this section, we present a few applications that use the FFT algorithm to compute unconditional expected allocations and observe their implications for risk-sharing. We start with a small portfolio of risks, where the FFT algorithm is not essential but will explain the method and point out numerical considerations. Then, we consider a larger portfolio to show that the method scales well to problems with many agents. 
We examine the numerical comparison of direct, recursive, and transform-based approaches and apply arithmetization techniques to a problem involving heavy-tailed risks. To the best of our knowledge, our method is the first to efficiently handle the large and heavy-tailed portfolios examined in this section unless each risk is identically distributed, testifying to the utility of our approach in practical situations. 

\subsection{Small portfolio of independent compound Poisson distributed rvs}\label{ss:app-small}

We replicate Case 1 of the application in Section 6.1 of \cite{denuit2019sizebiased}. Consider four participants in a pool, and each participant contributes risk $X_i$ that follows a compound Poisson distribution, with parameter $\lambda_i$ and a discrete severity whose pmf is $f_{C_i}$ with support $\{1, 2, 3, 4\}$, 
for $i \in \{1,\ldots,4\}$. 
We present the values of $\lambda_i$ and $f_{C_i}$ for each participant $i \in \{1,\ldots,4\}$ in Table \ref{tab:param-app-1}. 
\begin{table}
	\centering
	\begin{tabular}{cccccc}
		$i$ & $\lambda_i$ & $f_{C_i}(1)$ & $f_{C_i}(2)$ & $f_{C_i}(3)$ & $f_{C_i}(4)$ \\ \hline
		1  &    0.08     &     0.1      &     0.2      &     0.4      &     0.3      \\
		2  &    0.08     &     0.15     &     0.25     &     0.3      &     0.3      \\
		3  &     0.1     &     0.1      &     0.2      &     0.3      &     0.4      \\
		4  &     0.1     &     0.15     &     0.25     &     0.3      &     0.3
	\end{tabular}
	\caption{Values of $\lambda_i$ and $f_{C_i}$ for each participant $i \in \{1,\ldots,4\}$ for a small pool of four participants.}\label{tab:param-app-1}
\end{table}
We provide the \textsf{R} code in Appendix \ref{app:small}, the numerical values that follow come from \textsf{R} version 4.0.4. Besides the setup and validation code, the actual computation of conditional means takes fewer than 15 lines (even if the number of participants grows). We recover the values in \cite{denuit2019sizebiased}. 

In Figure \ref{fig:plots-app-1}, we present three graphs: the pmf of $S$, the total unconditional expected allocations for a given outcome of $S$, and the total conditional means for a given outcome of $S$. One should have $\sum_{i = 1}^{n}E[X_i \times 1_{\{S = k\}}] = k \Pr(S = k)$, which is what we observe in the middle plane of Figure \ref{fig:plots-app-1}.

To compute the unconditional expected allocations, we must divide the total unconditional expected allocations by the pmf of $S$. If the pmf of $S$ is very small for some values of $k$, the unconditional expected allocations may be inaccurate. This occurs since machines have finite precision, and exact zeroes are not always represented accurately due to underflow issues. The FFT algorithm introduces small numerical errors during the computation of the pmf of $S$ and the unconditional expected allocations. Still, these errors are negligible since they occur for events whose probability is close to the machine precision. Therefore, one should only consider the unconditional expected allocations for values of $k$ where $\Pr(S = k)$ is not too small, which are the values of $k$ where the unconditional expected allocations are accurate. 

While the transform-based approach proposed in this paper provides accurate unconditional expected allocations for the important values of $k$, it is useful to identify the values of $k$ where the unconditional expected allocations are inaccurate. One way to validate the accuracy of unconditional expected allocations is to validate that the full allocation property in \eqref{eq:full-allocation} holds. In this case, we expect $\sum_{i = 1}^{n} E[X_i \vert S = k] = k$ for all $k \in \mathbb{N}$. We plot the curve $\sum_{i = 1}^{n} E[X_i \vert S = k]$ in the right plane of Figure \ref{fig:plots-app-1}. That curve is linear between $k = 0$ and $k = 37$. However, one has $\sum_{i = 1}^{n} E[X_i \mid S = 38] = 38.05$, which is slightly higher than $38$. The FFT method of computing unconditional expected allocations provides inaccurate values when the mass function is under machine precision; for example, we have $\sum_{i = 1}^{n} E[X_i \mid S = 43] = 116$ and $\sum_{i = 1}^{n} E[X_i \mid S = 63] = -146$. However, we have $\Pr(S = 43) = 1.7\times 10^{-17}$ and $\Pr(S = 63) = 3.3\times 10^{-19}$, that is, they are numerically indecipherable from zero because of underflow. We will investigate the impact of underflow on the unconditional expected allocations in Section \ref{ss:application3a}, where we show that the FFT algorithm is still more useful than the direct and recursive methods for large portfolios of risks.

\begin{figure}
	\centering
	\includegraphics{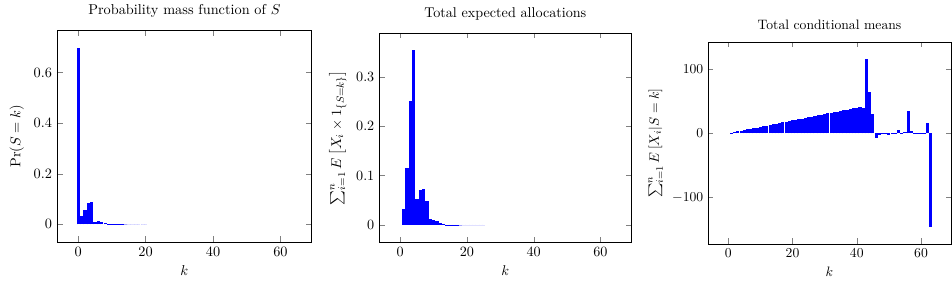}
	\caption{Left: pmf of $S$. Middle: $\sum_{i = 1}^{n}E[X_i \times 1_{\{S = k\}}]$. Right: $\sum_{i = 1}^{n}E[X_i \vert S = k]$.  }\label{fig:plots-app-1}
\end{figure}

\subsection{Large portfolio of independent compound Poisson distributed rvs}

We consider a portfolio or pool of 10{,}000 risks in the second application. Each risk $X_i$ is independent and follows a compound Poisson distribution with parameter $\lambda_i$, with severity rv $B_i \sim NB(r_i, q_i)$, implying that $E[X_i] = \lambda_i r_i (1-q_i)/q_i$, for $i = 1, \dots, 10{,}000$. We set each risk to have different triplets of parameters. For illustration purposes we simulate the triplets of parameters $(\lambda_i, r_i, q_i)$ for $i = 1, \dots, 10{,}000$ according to $\lambda_i \sim Exp(10)$, $r_i \sim Unif(\{1, 2, 3, 4, 5, 6\})$ and $q_i \sim Unif([0.4, 0.5])$ such that on average, $\lambda_i = 0.1$, $r_i = 3.5$ and $q_i = 0.45$. We present the simulated parameters and expected values for the first eight contracts in Table \ref{tab:first-param}.
\begin{table}[ht]
	\centering
	\begin{tabular}{rrrrrrrrr}
		$i$ &        1 &        2 &        3 &        4 &        5 &        6 &        7 &        8 \\ \hline
		$\lambda_i$ & 0.161152 & 0.031859 & 0.027368 & 0.238748 & 0.115137 & 0.470203 & 0.146247 & 0.011747 \\
		$q_i$ & 0.489756 & 0.423367 & 0.455898 & 0.451500 & 0.486834 & 0.440405 & 0.440082 & 0.481335 \\
		$r_i$ &        2 &        6 &        1 &        4 &        6 &        5 &        3 &        1 \\\hline
		$E[X_i]$ & 0.335788 & 0.260354 & 0.032662 & 1.160162 & 0.728190 & 2.987289 & 0.558214 & 0.012658
	\end{tabular}
	\caption{First eight sets of parameters.}
	\label{tab:first-param}
\end{table}

We present the code in Appendix \ref{app:large}, using \textsf{R} version 4.0.4. 
That code computes conditional means for each of the 10{,}000 unique risks and takes approximately 16 seconds on a personal computer (with a Intel\circledR Core\texttrademark i5-7600K CPU @ 3.80GHz CPU). 

In Figure \ref{fig:distn-cond-mean} we present the pmf of the conditional means $E[X_i \vert S]$, for $i \in \{1, \dots, 8\}$. Note that all eight pmfs share the same values on the $y$-axis but differ on the $x$-axis. This is because the relationship giving the probability masses for conditional means is 
$$\Pr(E[X_i | S] = E[X_i | S = k]) = \sum_{\{j \in \mathbb{N} : E[X_i | S = j] = E[X_i | S = k]\}}\Pr(S = j),$$
for all $i \in \{1, \dots, 10{,}000\}$ and $k \in \mathbb{N}$. In the case where each risk is heterogeneous, we have that $\Pr(E[X_i | S] = E[X_i | S = k]) \approx \Pr(S = k)$; the only difference is the domains of $E[X_i | S]$ for $i \in \{1, \dots, 10{,}000\}$. Indeed, although the pmf for the conditional means of risk $j = 2$ (in red) appears to be a single point mass, we observe by magnifying that the pmf shares the probability values from the other pmfs. Also, as shown by the authors of \cite{denuit2021risk} under mild technical conditions, the conditional means converge to the expected value. For illustration purposes, we add vertical dashed lines at the expected values. 

\begin{figure}[ht]
	\centering
	\includegraphics{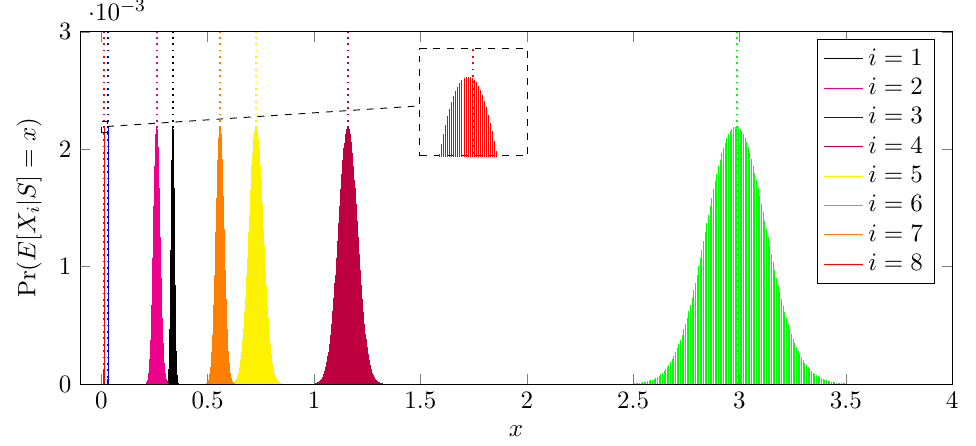}
	\caption{Probability mass function of the conditional means $E[X_i \vert S]$, for the contracts $i \in \{1,\dots,8\}$.
		The vertical lines go through the values $E[X_i]$, for $i \in \{1, \dots,8\}$. }\label{fig:distn-cond-mean}
\end{figure}

\subsection{Numerical comparison of approaches}\label{ss:application3a}

In this section, we compare the time required and the accuracy of the numerical approximation from the FFT algorithm. We will compare the direct approach, the recursive approach, and the FFT algorithm. To do this, we consider a portfolio of two risks $X_1$ and $X_2$ 
with compound Poisson distributions and Pareto severities, where the Poisson rates are respectively $\lambda_1 = 0.05$ and $\lambda_2 = 0.02$, the Pareto shapes are respectively $\alpha_1 = 4$ and $\alpha_2 = 5$ and the Pareto scales are respectively $\beta_1 = 2000$ and $\beta_2 = 3000$. The expected values are $E[X_1] = 33.3333$ and $E[X_2] = 15$. Since our approach works with discrete random variables, we consider discretized Pareto severities with different discretization steps using the moment matching method; see, for instance, Appendix E.2 of \cite{klugman2018loss} for discretization methods. Selecting a small discretization step will yield more accurate approximations but require more computation time. 

We discretize up to the truncation point of $k_{max} = 2^{18} = 262144$, such that, numerically, $F_{B_1}(k_{max}) = F_{B_2}(k_{max}) = 1$, so there is no aliasing error (up to machine precision). In Table \ref{tab:compute}, we present the computation time for the direct, recursive, and FFT approaches for different discretization steps $h$. We observe that the FFT approach is significantly faster than the direct and recursive approaches, and the time increase is more significant for the direct and recursive approaches. The reason is that the direct approach has a complexity of $O(1/h^3)$, the recursive approach has a complexity of $O(1/h^2)$, and the FFT approach has a complexity of $O(-\log(h)/h)$. In this case, the FFT approach enables one to compute the conditional means at the level of a penny in less than a minute, which is not possible with the direct and recursive approaches in a reasonable time. 

\begin{table}[ht]
	\centering
	\caption{Computation time for direct, recursive and FFT approaches.}\label{tab:compute}
	\begin{tabular}{lrrrrrr}
		$h$       &   1000 &     100 &       10 &        1 &   0.1 &   0.01 \\ \hline
		Direct    & 0.1446 & 15.4318 & 1987.458 &       -- &    -- &     -- \\
		Recursive & 0.0253 &   0.970 &   94.300 & 9766.796 &    -- &     -- \\
		Transform & 0.0003 &   0.003 &    0.025 &    0.238 & 5.835 & 59.955
	\end{tabular}
\end{table}

We will now study the accuracy of the numerical approximation for the direct, recursive, and FFT approaches. To allow for exact computation of the conditional means and to compare the approaches with the exact values, we consider a discretization step of $h = 100$. Denote by $g^{D}_{i:k}$, $g^{R}_{i:k}$, and $g^{T}_{i:k}$ the conditional means of $X_i$ given $S = k$ computed using the direct, recursive, and FFT approaches, respectively. The conditional mean risk-sharing rules satisfy the full allocation property \eqref{eq:full-allocation}. We compute the stopping time based on the transform-based conditional means
$$\gamma^*(k_{max}) = \inf \{k \in h\mathbb{N}: |g^{T}_{1:k} + g^{T}_{2:k} - k| \leq \varepsilon\},$$
where the conditional means are computed with marginals truncated at $k_{max}$. We assume that the conditional means are accurate up to the tolerance $\varepsilon$ for $k \leq \gamma^*$, and calculate the errors based on this assumption. The direct approach is exact up to the truncation point, that is, for $k_{max} > \gamma^*(k_{max})$, or if the truncation point is such that the probabilities sum to one (at least up to machine precision). The recursive and FFT approaches are approximations due to wrap-around, aliasing, and FFT underflow errors. 

To compute the errors of different computation methods, we consider two error measures: the sum of absolute errors and the supremum of absolute errors. We compute the errors for the direct, recursive and transform approaches with different truncation points. We will compute the exact conditional means using the direct approach with a truncation point of $2^{18}$, and we will denote the exact conditional means by $g^{*}_{i:k}$. The two error measures we consider are the sum of absolute errors and the supremum of absolute errors, defined as
$$E_{\text{sum}, j}^{m} := \sum_{k \in \{0, h, 2h, \dots, \gamma^*\}} 
\left| g^{*}_{j:h} - g^{m}_{j:k} \right|
$$
and
$$E_{\text{sup}, j}^{m} := \sup_{k \in \{0, h, 2h, \dots, \gamma^*\}} 
\left| g^{*}_{j:h} - g^{m}_{j:k} \right|,
$$
for $j\in\{1, 2\}$ and $m \in \{D, R, T\}$. We present, in Table \ref{tab:errors}, the errors for the direct, recursive, and transform approaches with different truncation points. When the truncation point is too small (e.g., $2^{14}$), the discretization probabilities for the severity distributions do not sum to one, which leads to errors in all approaches (mostly due to errors from computing the probability mass of $X_1$ and $X_2$). Note that, since the balance property is satisfied (up to $\varepsilon$), we will have $E_{\text{sum}, 1}^{m} \approx E_{\text{sum}, 2}^{m}$ and $E_{\text{sup}, 1}^{m} \approx E_{\text{sup}, 2}^{m}$, and we only report the errors for $X_1$ in Table \ref{tab:errors}. 
\begin{table}[ht]
	\centering
	\caption{Errors for the recursive and FFT approaches with different truncation points.}\label{tab:errors}
	\begin{tabular}{rrrrr}
		$k_{max}$ &  $2^{14} = 16384$ &  $2^{15}= 32768$ &  $2^{16}=65536$ & $2^{17}=131072$  \\ [0.3cm] \hline 
$F_{B_1}(k_{max})$ & 0.9998599 & 0.9999891 & 0.9999992 & 0.9999999 \\
$F_{B_2}(k_{max})$ & 0.9999112 & 0.9999958 & 0.9999998 & 1.0000000 \\
$E_{\text{sum}, 1}^{D}$ & 91420.947 &  783.3588 &  0.000000 & 0.000000  \\
$E_{\text{sup}, 1}^{D}$ &  2465.381 &  783.3588 &  0.000000 & 0.000000  \\[0.3cm]
$E_{\text{sum}, 1}^{R}$ & 35620.298 &  239.3113 &  0.541361 & 1.245401  \\
$E_{\text{sup}, 1}^{R}$ &  1265.880 &  239.2936 &  0.017958 & 0.027168  \\[0.3cm]
$E_{\text{sum}, 1}^{T}$ & 35620.312 &  239.3079 &  0.536580 & 1.243309  \\
$E_{\text{sup}, 1}^{T}$ &  1265.880 &  239.2937 &  0.018011 & 0.027106  \\[0.3cm]
$\gamma^*(k_{max})$ &     20800 &     32800 &     56800 & 64000     \\
$F_{S}(\gamma^*(k_{max}))$ & 0.9999963 & 0.9999994 & 0.9999999 & 1
\end{tabular}
\end{table}

When the truncation point $k_{max}$ is too small, $F_{B_1}(k_{max})$ and $F_{B_2}(k_{max})$ are smaller than one and we have $\gamma^*(k_{max}) > k_{max}$, leading to errors in all approaches (the largest errors coming from the direct approach). In this application, we have for $k_{max} = 2^{15}$ that $\gamma^*(k_{max})$ larger than $k_{max}$ for a single step, leading to large errors. For $k_{max} = 2^{16}$ and $k_{max} = 2^{17}$, we have $\gamma^*(k_{max}) < k_{max}$, so there are no errors for the direct approach. The wrap-around errors from the recursive and FFT approaches are similar. While the errors are larger for $k_{max} = 2^{17}$ than for $k_{max} = 2^{16}$, more values are computed, and the extra values are computed for a higher $k$ where the errors are larger. Still, the errors are small, and the worst-case error corresponds to three pennies for the recursive and FFT approaches, which is acceptable for most applications. Since the FFT approach is significantly faster than the recursive approach, we recommend using the FFT approach for computing the conditional means.

\subsection{Portfolio of heavy tailed risks}\label{ss:application3}

Next, we consider the computation of unconditional expected allocations for a portfolio of heavy-tailed risks. In particular, we consider risks whose variance does not exist; hence, the central limit theorem results of \cite{denuit2021risk} do not hold because the variance of the sum of each rv does not exist. We consider a portfolio of size $n \in \{3, 100, 1000\}$ and compare the behaviour of the first three contracts. Our goal is to illustrate empirically that the conditional mean for each contract converges to their marginal mean. We set $X_i, i \in \{1, \dots, n\}$, to follow an arithmetized Pareto distribution defined using the moment matching method. Further, we select parameters $\alpha_i \in [1.3, 1.9]$, for $i \in \{1, \dots, n\}$, so the variance of individual risks does not exist. For the first three risks, we select $(\alpha_1, \alpha_2, \alpha_3) = (1.3, 1.6, 1.9)$ and $(\lambda_1, \lambda_2, \lambda_3) = (10(1.3 - 1), 10(1.6 - 1), 10(1.9 - 1))$ such that $E[X_i] \approx 10$ for $i \in \{1, 2, 3\}$. We write the approximate symbol since the mean may not be preserved exactly due to truncation since Pareto rvs are heavy-tailed. For the remaining risks $X_i, i \in \{4, \dots, 1000\}$, we simulate the parameters according to $\alpha_i \sim Unif([1.3, 1.9])$ and $\lambda_i \sim Unif([5, 15])$, implying $50/9 \leq E[X_i] \leq 50$ for $i \in \{4, \dots, 1000\}$, and the variance does not exist for any risk in the portfolio. We provide the 
\textsf{R}
code in Appendix \ref{app:heavy}.

In Figure \ref{eq:cdf-pareto}, we present the cdf of the conditional means for risks $X_1, X_2$ and $X_3$. The dashed, dotted, and dash-dotted lines present the cdf of conditional means for $n = 3, 100$ and 1000 respectively. Due to the heavy-tailed risks, one must select a large truncation point $k_{max}$ to avoid aliasing (see, for instance, \cite{grubel1999computation} and \cite{embrechts2009panjer} for discussions on aliasing with FFT methods for aggregation). Hence, we compute $1000 \times 2^{20}$ values, which takes approximately 9 minutes on a personal laptop. To facilitate comparisons, we present the cdf of $X_i$ ($n=1$) in black and the expected value of $X_i$ in green (vertical line), $i \in \{1, 2, 3\}$. Each cdf crosses once. Therefore, according to the Karlin-Novikoff criteria, given that they share the same mean, the conditional means are ordered under the convex order, as expected; see, for instance, \cite{denuit2012convex}. One may observe that the cdfs of the conditional means approach the cdf of a degenerate rv at the mean. The conditional mean of $X_3$ approaches the degenerate rv at its mean faster since its tail is lighter than $X_1$ or $X_2$. Indeed, one observes that the cdf of $E[X_3 \vert S = x]$ is almost vertical, while the cdf of $E[X_1 \vert S = x]$ is not.
\begin{figure}[ht]
	\centering
	\includegraphics{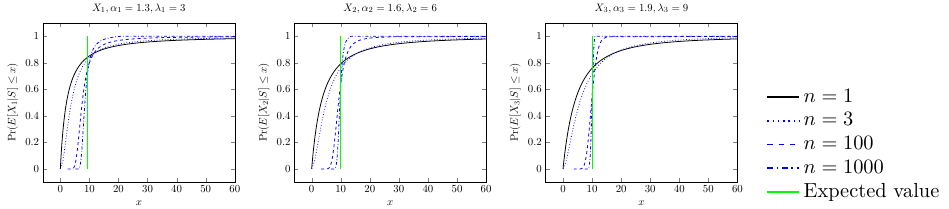}
	\caption{Cumulative distribution function of conditional means for $n = 1, 3, 100, 1000$.	
	}\label{eq:cdf-pareto}
\end{figure}

According to this application, one observes that the conditional mean $E[X_1 \vert S]$ converges in distribution to the expected value $E[X_1]$ as the portfolio size increases. However, future research remains to show that this conjecture is true in general, that is, providing a law of large numbers result for the conditional mean, generalizing the results of \cite{denuit2020largeloss} and \cite{denuit2021risk}.

\subsection{Small portfolio of heterogeneous losses}\label{sec:bern}

Let $\boldsymbol{I} = (I_1, \dots, I_n)$ be a vector of independent Bernoulli rvs with marginal probabilities $q_i \in (0, 1)$, for $i \in \{1, \dots, n\}$. Further define the rv $X_i = b_i \times I_i$, with $b_i \in \mathbb{N}_1$, for $i \in \{1, \dots, n\}$. This model is sometimes called the individual risk model (with a fixed payment amount) and has applications, for instance, in life insurance, where death benefits are usually known in advance, or for insurance-linked securities in situations where investors recover their initial investment unless a trigger event occurs before the maturity date. The interested reader may refer to \cite{klugman2018loss} for detailed examples of the individual risk model. The multivariate pgf of $\boldsymbol{X} = (X_1, \dots, X_n)$ is 
\begin{equation}\label{eq:pmf-S-bern}
    \mathcal{P}_{\boldsymbol{X}}(t_1, \dots, t_n) =\prod_{i = 1}^{n} (1 - q_i  + q_i t_i^{b_i}),
\end{equation}
while the OGF of the sequence of unconditional expected allocations for risk $X_1$ is
$$\mathcal{P}_{S}^{[1]}(t) = q_1 b_1t^{b_1}\prod_{i = 2}^{n} (1 - q_i + q_i t^{b_i}).$$
To compute exact values of the pmf and unconditional expected allocations using the FFT approach, one must select $k_{max} \geq 1 + \sum_{i = 1}^{n} b_i$ (or select the smallest $m$ such that $2^m \geq 1 + \sum_{i = 1}^{n} b_i$).

Let us discuss some of the theoretical difficulties with computing the conditional means in the context of this application. To do so, we will need some notation. The cardinality of a set $\mathcal{A}$ is denoted by $|\mathcal{A}|$. Define the set $\mathcal{B} = \{(x_1, \dots, x_n):x_i \in \{0, b_i\}, 1 \leq i \leq n\}$ as all distinct possible outcomes of $\boldsymbol{X}$. Note that $|\mathcal{B}| = 2^n$. Define $\mathcal{B}_k = \{(x_1, \dots, x_n) \in \mathcal{B}:
\sum_{i=1}^n x_i = k\}$, for $k = 0, 1, \ldots, s_{max}$, where $s_{max} = \sum_{i=1}^n b_i$.
Note that $|\mathcal{B}_{0}| = 1$ and $|\mathcal{B}_{s_{max}}| = 1$.

The sets $\mathcal{B}_{k}$ and $\mathcal{B}_{k'}$ are mutually exclusive, i.e. $\mathcal{B}_k \cap \mathcal{B}_{k'} = \varnothing$, for $k \neq k' \in \{1,\ldots,n\}$. Also, $\bigcup_{k=0}^{s_{max}} \mathcal{B}_k = \mathcal{B}$. When $\mathcal{B}_k$ is empty ($\mathcal{B}_k = \varnothing$), we have $|\mathcal{B}_k| = 0$, meaning the event $\{S=k\}$ is impossible. 
Such situations may occur when the number of contracts is small and the coverage amounts are heterogeneous. We say that $k$ is a possible outcome of the total losses $S$ if $|\mathcal{B}_{k}| > 0$. 

Fix $k \in \{0,1,\dots,s_{max}\}$ such that $|\mathcal{B}_{k}| = 1$, and let $(x_1,\dots,x_n)$ be the element of that singleton. This implies that the conditional expectation is given by $E[X_i | S = k] = x_i$, for $x_i \in \{0,b_i\}$, which means that the support of $E[X_i | S]$ is $0$ or its full coverage $b_i$, for $i \in \{1,\dots,n\}$. In other words, the support of $E[X_i | S ]$ is the same as the support of $X_i$; and a participant in a pool has not benefited from a diversification of its risk, no matter the size $n$ of the portfolio. As $|\mathcal{B}_{k}|$ increases, the support of $E[X_i | S ]$ has more elements, and these are the situations where insurance provides more value to customers. Counting the number of partitions of a set is a difficult problem in number theory. Fortunately, the OGF method provides a numerical solution to compute the unconditional expected allocations without further notions of number theory. See also Example 4.1 of \cite{denuit2021mortality} for a situation where some participants do not diversify due to partitions of odd numbers.

We consider a portfolio of $n = 6$ risks. 
We present the parameters for this example in Table \ref{tab:ex-bern-indep}, and the code to replicate this study is in Appendix \ref{app:arch}.
\begin{table}[ht]
	\centering
	\begin{tabular}{ccccccc}
		$i$   &  1   &  2  &  3   &  4  &  5   &  6  \\\hline 
		$b_i$ &  1   &  3  &  10  &  4  &  5   & 10  \\
		$q_i$ & 0.8  & 0.2 & 0.3 & 0.05 & 0.15 & 0.25
	\end{tabular}
	\caption{Marginal parameters for a small portfolio of heterogeneous losses.}\label{tab:ex-bern-indep}
\end{table}

In Figure \ref{fig:bern-cond-means-indep}, we present the conditional means along with the pmf and cdf of conditional means for risks $X_1$, $X_2$ and $X_3$. We describe each panel in the following:

\begin{itemize}
	\item The left panel presents the values of $E[X_i \vert S]$, for $i \in \{1, 2, 3\}$. Note that for the claim severity values in Table \ref{tab:ex-bern-indep}, we have $|\mathcal{B}_{2}| = |\mathcal{B}_{31}| = 0$, hence the events $S = 2$ and $S = 31$ are impossible and we have $\Pr(S = 2) = \Pr(S = 31) = 0$. When computed using the pgf in \eqref{eq:pmf-S-bern-archimedian} and the FFT algorithm, we have $\Pr(S = 2) = \Pr(S = 31) \approx 10^{-16}$ since this is the underflow error using double precision with IEEE 754. Hence, the conditional means should be 0 for $k = 2$ and $k = 31$; dividing two underflowed values generates erratic results. These values should be rejected from the analysis, but we show them in red as a warning of numerical problems with the FFT method if one is not wary of underflow versus true zeroes when using the FFT method. As in other applications, one should observe the total conditional means (row 4 of Figure \ref{fig:bern-cond-means-indep}) and retain the values that form a step function with steps of 1. Conditional means that deviate from their expected total should be discarded due to underflow or division by zero. However, the events which cause numerical issues have zero or negligible probability (under $10^{-16}$); hence, the expectations of interest do not suffer from underflow. 
	
	Also of interest is the shape of conditional means as a function of $k$. For $i = 1$, we have unpredictable unconditional expected allocations since the outcome 1 is often a part of $\mathcal{B}_{k}, k \in \mathbb{N}_1$, and $q_1$ is greater than $q_i, i \in \{2, \dots, 6\}$. For $i = 2$, we have predictable unconditional expected allocations since $3$ is a part of $\mathcal{B}_k$ for cyclical values of $k$, and $3$ does not divide the other values of $b_i, i \in \{1, 3, 4, 5, 6\}$. Finally, we have $b_3 = 10$ and $b_6 = 10$. Hence, the conditional allocations are often shared between risks $X_3$ and $X_6$, though not perfectly since $q_3 \neq q_6$. 
    
    In row 3 of Figure \ref{fig:bern-cond-means-indep}, we have a mass around 5.6 since the outcomes $(X_1 = 1, X_3 = 0, X_4 = 4, X_5 = 5, X_6 = 0)$ and $(X_1 = 0, X_3 = 0, X_4 = 0, X_5 = 0, X_6 = 10)$ also yields $S = 10$, so these events diversify the event $(X_1 = 0, X_3 = 10, X_4 = 0, X_5 = 0, X_6 = 0)$. 
	\item The middle panel presents the pmf of $E[X_i | S]$, for $i \in \{1, 2, 3\}$. The support of this rv is the set of values $\{E[X_i \vert S = k], k \in \mathbb{N}\}$, for $i \in \{1, 2, 3\}$. Notice that the support of this rv is sparse for small portfolios with heterogeneous values of $b_i, i \in \{1, \dots, n\}$.
	\item The right panel presents the cdf of $E[X_i \vert S]$, for $i \in \{1, 2, 3\}$, simplifying the interpretation of the middle panel since the probability masses may appear close together. 
\end{itemize}

\begin{figure}[H]
	\centering
	\includegraphics[width = 0.7\textwidth]{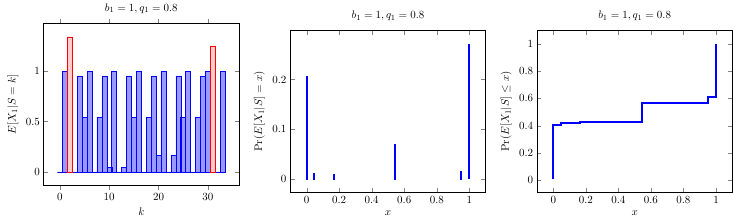}
	\includegraphics[width = 0.7\textwidth]{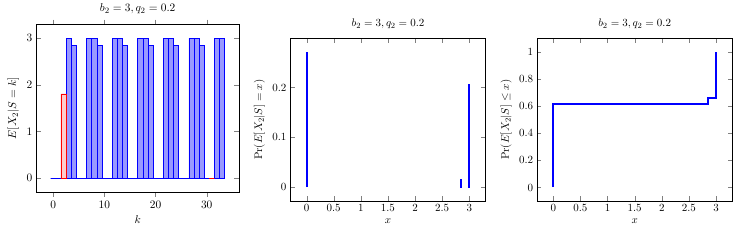}
	\includegraphics[width = 0.7\textwidth]{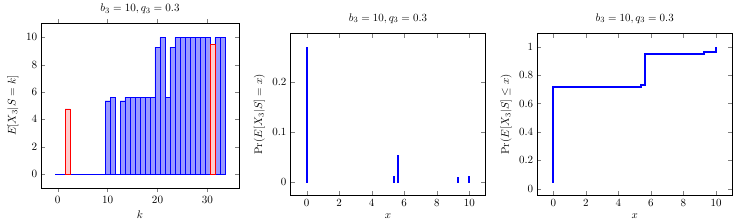}
	\includegraphics[width = 0.7\textwidth]{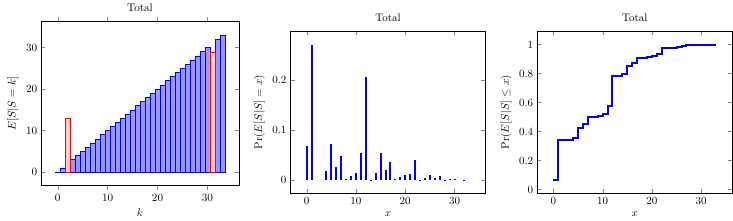}
	\caption{Left: conditional means. Middle: pmf of conditional means. Right: cdf of conditional means.}
	\label{fig:bern-cond-means-indep}
\end{figure}

Note that as more participants enter the pool, more risks may diversify; that is, $\mathcal{B}_{k}$ has a higher cardinality for all $k \in \mathbb{N}_1$. The risks diversify, and the pmf of unconditional expected allocations is less sparse. In Figure \ref{fig:bern-cond-means-50-indep}, we replicate the above study but add 69 participants where we sample the parameters according to $q_i \sim Unif([0, 1])$ and $b_i \sim Unif(\{1, 2, \dots, 10\})$. We present the results for risk 3 (with $b_3 = 10)$ and the total pool in Figure \ref{fig:bern-cond-means-50-indep}. Once again, we observe numerical issues for large values of $k$ in the left panel. However, the middle panel is much less sparse than in Figure \ref{fig:bern-cond-means-indep}. 
\begin{figure}[ht]
	\centering
	\includegraphics[width = 0.9\textwidth]{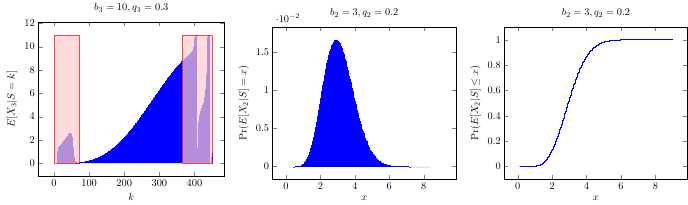}
	\includegraphics[width = 0.9\textwidth]{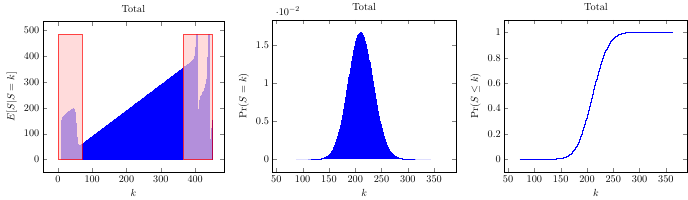}
	\caption{Pool of 75 participants. Left: conditional means. Middle: pmf of conditional means. Right: cdf of conditional means.}\label{fig:bern-cond-means-50-indep}
\end{figure}

\section{Sum of dependent rvs}\label{sec:dependent}

One may also use the methods described in this paper to compute unconditional expected allocations for dependent rvs. One obtains convenient results when the multivariate pgf is simple to differentiate, which is sometimes the case for mixture models (which include common shock models). The results from this section supplement the literature on risk allocation or risk sharing for mixture models as studied in Section 3 of \cite{cossette2018dependent} or Section 4 of \cite{denuit2021conditionala}.

\subsection{Multivariate Poisson distribution constructed with common shocks}

As a first example, we present a common shock model. Multivariate Poisson distributions based on common shocks are studied notably in \cite{teicher1954multivariate} and \cite{mahamunulu1967note}. The interested reader may also consult \cite{lindskog2003common} for actuarial applications of common shock Poisson models.

\begin{example}[Hierarchical common Poisson shocks]
	Let $Y_{A} \sim Pois(\lambda_A)$ for $A \in \{ \{1, 2\}^3 \cup  \{1, 2\}^2 \cup \{0, 1, 2\}\}$ be independent rvs. We construct dependent rvs through the common shock framework $X_{ijk} = Y_{ijk} + Y_{ij} + Y_{i} + Y_0$ for $(i, j, k) \in \{1, 2\}^3$. This is a special case of the multivariate Poisson distribution from \cite{mahamunulu1967note}, and we illustrate the dependence structure in Figure \ref{fig:pois-common-shock}.
	\begin{figure}[ht]
		\centering
		\includegraphics{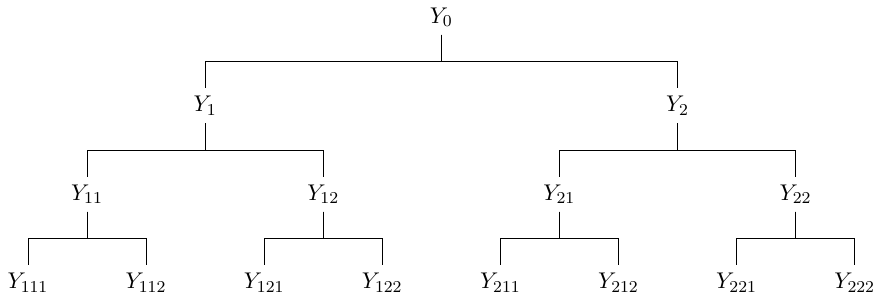}
	\caption{Hierarchical Poisson common shock structure.}\label{fig:pois-common-shock}
	\end{figure}
	Let $S = \sum_{(i, j, k) \in \{1, 2\}^3} X_{ijk}.$
	Then, one may verify that $S$ follows a compound Poisson distribution, so one may use Panjer recursion or FFT to compute the values of the pmf of $S$. Further, the OGF for the unconditional expected allocations of risk $X_{ijk}$, for $(i,j,k) \in \{1, 2\}^3$, is 
	$$\mathcal{P}_{S}^{[ijk]}(t) = \left(\lambda_{ijk}t + \lambda_{ij} t^2 + \lambda_{i} t^4 + \lambda_0 t^{8}\right)\mathcal{P}_S(t).$$	
	For $(i,j,k) \in \{1, 2\}^3$, we deduce that
	\begin{align*}
		E\left[X_{ijk} \times 1_{\{S = k\}}\right] &= \begin{cases}
			0, & k = 0\\
			\lambda_{ijk} f_S(k-1), & k = 1\\
			\lambda_{ijk} f_S(k-1) + \lambda_{ij} f_S(k-2), & k = 2, 3\\
			\lambda_{ijk} f_S(k-1) + \lambda_{ij} f_S(k-2) + \lambda_{i}f_S(k-4), & k = 4, \dots, 7\\
			\lambda_{ijk} f_S(k-1) + \lambda_{ij} f_S(k-2) + \lambda_{i}f_S(k-4) + \lambda_0 f_{S}(k-8), & k = 8, 9, \dots
		\end{cases}.
	\end{align*}
More general Poisson common shock models, as proposed in \cite{mahamunulu1967note}, yield similar expressions for expected and cumulative unconditional expected allocations. 
\end{example}

\subsection{Multivariate mixed Poisson distribution}

Next, we consider a multivariate mixed Poisson distribution. We induce dependence using a mixture random vector $\boldsymbol{\Theta} = (\Theta_1, \dots, \Theta_n)$ with $E[\Theta_i] = 1$ for $i = 1, \dots, n$. Consider a vector of conditionally independent rvs $(X_i \vert \Theta_i = \theta_i) \sim Poisson(\lambda_i \theta_i)$ for $i = 1, \dots, n$. 
The multivariate pgf of $(X_1, \dots, X_n)$ is
\begin{equation}\label{eq:pmf-mx-pois}
	\mathcal{P}_{\boldsymbol{X}}(t_1, \dots, t_n) = E_{\boldsymbol{\Theta}} \left[\e^{\Theta_1 \lambda_1 (t_1 - 1)} \dots \e^{\Theta_n \lambda_n (t_n - 1)}\right] = \mathcal{M}_{\boldsymbol{\Theta}}(\lambda_1(t_1 - 1), \dots, \lambda_n(t_n - 1)), 
\end{equation}
where $\mathcal{M}_{\boldsymbol{\Theta}}$ is the multivariate moment generating function (mgf) of $\boldsymbol{\Theta}$. 
Then, combining Theorem \ref{th:TheBigResult} and \eqref{eq:pmf-mx-pois}, 
we find that 
\begin{equation}\label{eq:agf-mx-pois}
	\mathcal{P}^{[1]}_{S}(t) = \left.\lambda_1 t \left[\frac{\partial }{\partial x}\mathcal{M}_{\boldsymbol{\Theta}}(x, \lambda_2(t - 1), \dots, \lambda_n(t - 1))\right]\right\vert_{x = \lambda_1(t-1)}.
\end{equation}

\begin{example}[Poisson-gamma common mixture]
	We consider a mixture distribution from a bivariate gamma common shock model described in \cite{mathai1991multivariate}. Let us define three independent rvs $Y_i, i \in \{0, 1, 2\}$ where $Y_0 \sim Gamma(\gamma_0, \beta_0)$, and $Y_i \sim Gamma(r_i - \gamma_0, r_i)$ for $i \in \{1, 2\}$ with $0 \leq \gamma_0 \leq \min(r_1, r_2)$. Let $\Theta_i = \beta_0 / r_i Y_0 + Y_i$ for $i = 1, 2$. Then the pair of rvs $(\Theta_1, \Theta_2)$ follows a bivariate gamma distribution with marginals $\Theta_i \sim Ga(r_i, r_i)$, $i = 1, 2$ and $\gamma_0$ is a dependence parameter. 
	The bivariate mgf of the pair of rvs $(\Theta_1,\Theta_2)$ is 
	\begin{equation}\label{eq:mgf-gamma-crmm}
		\mathcal{M}_{\Theta_1, \Theta_2}(x_1, x_2) = \left(1 -\frac{x_1}{r_1}\right) ^{-(r_1 -\gamma_0)} \left(1-\frac{x_2}{r_2}\right)^{-(r_2 - \gamma_0)} \left(1 - \frac{x_1}{r_1} - \frac{x_2}{r_2}\right)^{-\gamma_0}
	\end{equation}
	and its derivative with respect to $x_1$ is
	\begin{equation}\label{eq:mgf-gamma-crmm-derivee}
		\frac{\partial }{\partial x_1} \mathcal{M}_{\Theta_1, \Theta_2}(x_1, x_2) = \left( \frac{r_1 - \gamma_0}{r_1}\frac{1}{1-x_1/r_1} + \frac{\gamma_0}{r_1}
		\frac{1}{1 - x_1/r_1 - x_2 / r_2} \right)\mathcal{M}_{\Theta_1, \Theta_2}(x_1, x_2).
	\end{equation} 
	Consequently, the mixed Poisson distributed random vector $(X_1, X_2)$ follows a bivariate negative binomial distribution. It follows from \eqref{eq:pmf-mx-pois} and \eqref{eq:mgf-gamma-crmm} that
	$$\mathcal{P}_{S}(t) = \left(1 - \zeta_1(t - 1)\right)^{-(r_1 - \gamma_0)}\left(1-\zeta_2(t - 1)\right)^{-(r_2 - \gamma_0)} \left(1 - \zeta_{12} (t - 1)\right)^{-\gamma_0},$$
	where $\zeta_1 = \lambda_1/r_1$, $\zeta_2 = \lambda_2/r_2$ and $\zeta_{12} = \lambda_1/r_1 + \lambda_2/r_2$. We recognize that $S$ is the sum of three independent negative binomial rvs with parameters $(r_1 - \gamma_0, 1/(1-\zeta_1))$, $(r_2 - \gamma_0, 1/(1-\zeta_2))$ and $(\gamma_0, 1/(1-\zeta_{12}))$. The expression of the pmf $f_S$ of $S$ is given in Theorem 1 of \cite{furman2007convolution}. From \eqref{eq:agf-mx-pois} and \eqref{eq:mgf-gamma-crmm-derivee}, we get the following expression for the OGF for unconditional expected allocations:
	$$\mathcal{P}^{[1]}_{S}(t) = \lambda_1 t \left( \frac{1 - \gamma_0/r_1}{1 - \zeta_1(t-1)} + \frac{\gamma_0/r_1}{1 - \zeta_{12}(t-1)} \right)	\mathcal{P}_{S}(t).$$
	Finally, we can recover the unconditional expected allocations using the FFT algorithm or with the recursive-type formula 
	$$[t^k]\mathcal{P}^{[1]}_S(t) = E\left[X_1 \times 1_{\{S = k\}}\right] = \lambda_1 \sum_{j = 0}^{k-1}\left[\left(1-\frac{\gamma_0}{r_1}\right)\frac{1}{1 + \zeta_1}\left(\frac{\zeta_1}{1 + \zeta_1}\right)^{j} + \frac{\gamma_0}{r_1}\frac{1}{1 + \zeta_{12}}\left(\frac{\zeta_{12}}{1 + \zeta_{12}}\right)^{j}\right]f_{S}(k-1-j).$$
	One may develop similar expressions for cumulative unconditional expected allocations, applying the cumulative operator to the geometric series or to the pmf of $S$. 
\end{example}

\subsection{Multivariate Bernoulli distributions defined with Archimedean copulas}

Finally, we consider a multivariate Bernoulli distribution whose dependence structure is defined with an Archimedean copula. 
Let $(I_1, \dots, I_n)$ form a random vector, where the marginal distributions are Bernoulli with success probability $q_i \in (0, 1)$, for $i \in \{1, \dots, n\}$. Following \cite{marshall1988families}, we define the random vector according to $\Pr(I_i = 1 \vert \Theta = \theta) = r_i ^{\theta}$, where $\Theta$ is a mixing rv with a distribution defined on a strictly positive support. The relationship between the parameters $r_i$ and $q_i$ is
$$\Pr(I_i = 1) = E_{\Theta}\left[r_i ^{\Theta}\right] = \mathcal{L}_{\Theta}(-\ln r_i),$$
from which it follows that $r_i = \exp\{-\mathcal{L}_{\Theta}^{-1}(q_i)\}$, where
$\mathcal{L}_{\Theta}(t)$ and $\mathcal{L}_{\Theta}^{-1}(t)$ are respectively the Laplace-Stieltjes transform and the inverse Laplace-Stieltjes transform of the mixing rv. Further define the rv $X_i = b_i \times I_i$, with $b_i \in \mathbb{N}_1$ for $i \in \{1, \dots, n\}$. Note that the rvs $(X_i | \Theta = \theta)$ are conditionally independent, for $i \in \{1, \dots, n\}$ and $\theta > 0$. It follows that the multivariate pgf of $\boldsymbol{X} = (X_1, \dots, X_n)$ is 
$$\mathcal{P}_{\boldsymbol{X}}(t_1, \dots, t_n) = E\left[\prod_{i = 1}^{n} (1 - r_i ^{\Theta} + r_i ^{\Theta} t_i^{b_i})\right] = \int_{0}^{\infty} \prod_{i = 1}^{n} (1 - r_i ^{\theta} + r_i ^{\theta} t_i^{b_i}) \, \mathrm{d} F_{\Theta}(\theta).$$

We note that the underlying dependence structure in this model is an Archimedean copula; see, for instance, \cite{marshall1988families}, Section 4.7.5.2 of \cite{denuit2006actuarial} or Section 7.4 of \cite{mcneil2015quantitative} for the frailty construction of Archimedean copulas using common mixtures.

We consider the case where $\Theta$ is a discrete rv with support $\mathbb{N}_1$. Following the computational strategy from \cite{cossette2018dependent}, we select a threshold value $\theta^* = F_{\Theta}^{-1}(1 - \varepsilon)$ for a small $\varepsilon > 0$ and we have
\begin{equation}\label{eq:pmf-S-bern-archimedian}
	\mathcal{P}_{S}(t) = \sum_{\theta = 1}^{\theta^*}\Pr(\Theta = \theta)\prod_{i = 1}^{n} (1 - r_i ^{\theta} + r_i ^{\theta} t^{b_i}).
\end{equation}
Note that when the components of the random vector are independent, the rv $S$ follows a generalized Poisson-binomial distribution \cite{zhang2018generalized}. In the case of \eqref{eq:pmf-S-bern-archimedian}, we notice the pgf of a mixture of generalized Poisson-binomial distributions, where the mixture rv comes from the frailty construction of Archimedean copulas.

The OGF of the sequence of unconditional expected allocations for risk $X_1$ is
$$\mathcal{P}_{S}^{[1]}(t) = \sum_{\theta = 1}^{\theta^*}\Pr(\Theta = \theta)r_1 ^{\theta} b_1t^{b_1}\prod_{i = 2}^{n} (1 - r_i ^{\theta} + r_i ^{\theta} t^{b_i}).$$

\begin{example}\label{exa:bern}
We consider a portfolio of $n = 6$ risks, with $\Theta$ following a shifted geometric rv with pmf $f_{\Theta}(k) = (1-\alpha)\alpha^{k-1}$, for $k \in \mathbb{N}_1$, with $\alpha = 0.5$. It follows that the underlying dependence structure is an Ali-Mikhail-Haq copula. Following \cite{cossette2018dependent}, we select a threshold $\varepsilon = 10^{-10}$, such that $\theta^*= 34$. The indemnity payments are the same as in Table \ref{tab:ex-bern-indep}. We present the validation curve, the pmf for the conditional means of risk $X_3$, and the pmf of $S$ in Figure \ref{fig:bern-cond-means-arch} for $\alpha \in \{0, 0.1, 0.5, 0.8, 0.95\}.$ Increasing the dependence parameter increases the probability of zero contributions and full ($X_3 = b_3$) contributions. For other allocation values, the support of $E[X_3 | S]$ tends to cluster around the same value of 6 since increasing the dependence also increases the probability of mutual occurrence. Indeed, the probabilities for the outcomes $X_1 = 1, X_4 = 4$ and $X_5 = 5$ become more likely (resp. 0.006, 0.007, 0.011, 0.02 and 0.032 for $\alpha =$ 0, 0.1, 0.5, 0.8 and 0.95), so more diversification occurs when the total costs are divisible by 10, as $\alpha$ increases. 
\begin{figure}[ht]
	\centering
	\includegraphics{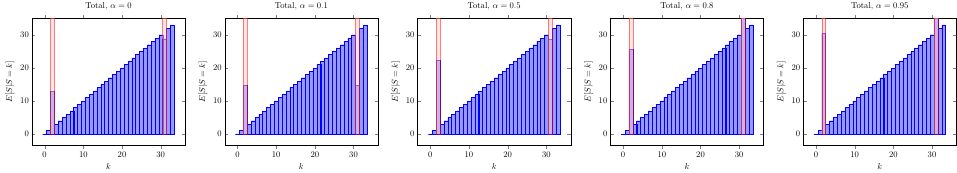}
	\includegraphics{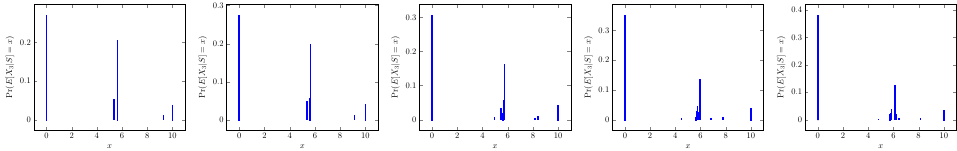}
	\includegraphics{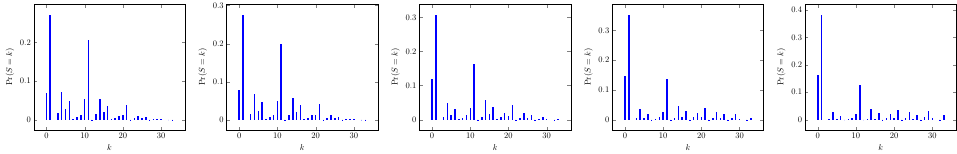}
    
	\caption{Pool of six participants.}\label{fig:bern-cond-means-arch}
\end{figure}

\begin{figure}[ht]
	\centering
	\includegraphics{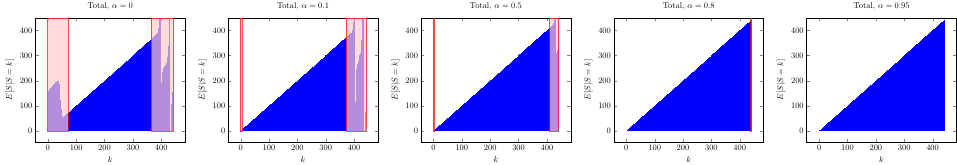}
	\includegraphics{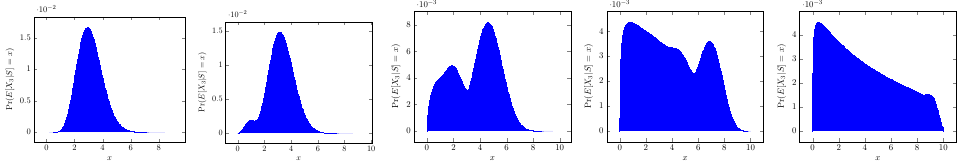}
	\includegraphics{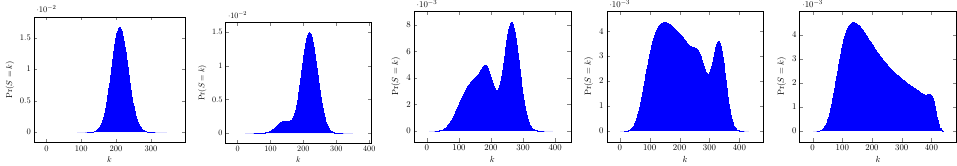}
	\caption{Pool of 75 participants.}\label{fig:bern-cond-means-50-arch}
\end{figure}

Next, we add 69 participants to the pool to investigate the effect of reducing the sparsity of the possible unconditional expected allocations. We present the validation curve, the values of $\Pr(E[X_3|S] = k)$ for $k \in \{0, \dots, 441\}$ and $\alpha = \{0, 0.1, 0.5, 0.8, 0.95\}$ in Figure \ref{fig:bern-cond-means-50-arch}. Note that the pmf of $S$ does not always converge to a normal distribution; hence, central limit theorems do not apply. Indeed, the common mixture representation of the Ali-Mikhail-Haq copula generates multiple nodes for the pmf of $S$ in this example. However, the OGF method with the FFT algorithm lets us easily extract the exact values of the pmf of unconditional expected allocations. As we increase the dependence parameter, the probability masses of $S$ and $E[X_3|S]$ are less concentrated around their means; thus, the tail of the distributions have non-zero mass, so there are no numerical issues in the validation curve. The code for this example is provided in Appendix \ref{app:arch}.
\end{example}

\section{Discussion}\label{sec:discussion}

We proposed a generating function method to compute the unconditional expected allocation, which has valuable applications in peer-to-peer insurance and risk allocation problems. The method simplifies solutions to risk allocation problems and enables FFT-based algorithms for fast computations.

The link between derivatives of pgfs and conditional distributions is not new. See, for instance, the use of derivatives to study conditional distributions with Poisson rvs \cite{subrahmaniam1966test, kocherlakota1988compounded} or with phase-type distributions \cite{ren2017cmph}. In a bivariate setting, \cite{kocherlakota1992bivariate} show that the conditional pgf of $X_1$ given the sum $S = X_1 + X_2 = s$ is
$$\mathcal{P}_{X_1 | S}(t_1 | k) = \frac{\left.\frac{\partial^k}{\partial t_2^k} \mathcal{P}_{X_1, X_2}(t_1t_2, t_2)\right\vert_{t_1 = t, t_2 = 0}}{\left.\frac{\partial^k}{\partial t_2^k} \mathcal{P}_{X_1, X_2}(t_1t_2, t_2)\right\vert_{t_1 = 1, t_2 = 0}},$$
for $|t_1| \leq 1$ and $k \in \mathbb{N}$. However, computing conditional expected values would involve computing multiple partial derivatives of the bivariate pgf. The method proposed in this paper only requires one partial derivative, a more convenient and tractable task. 

We remark that the generating function method provides a simpler proof of the size-biased transform method of computing unconditional expected allocations for discrete rvs. With $\widetilde{X}$ representing the size-biased transform of the rv $X$, along with the definition of the size-biased transform in \eqref{eq:pmf-size-biased}, the relationship between the pgfs of $X$ and $\widetilde{X}$ is
	$$\mathcal{P}_{\widetilde{X}}(t) = E\left[t^{\widetilde{X}}\right] = \sum_{k = 0}^{\infty} t^{k}f_{\widetilde{X}}(k) = \frac{t}{E[X]}\sum_{k = 0}^{\infty} k t^{k-1}f_{X}(k) = \frac{t}{E[X]}\frac{\diff}{\diff t} \sum_{k = 0}^{\infty} t^{k}f_{X}(k) = \frac{t}{E[X]} P_{X}'(t),$$
	for $|t| \leq 1$. Alternatively, one can obtain the pgf of $\widetilde{X}$ by applying operation \ref{prop:rightshift} (right shift) and \ref{prop:index-multiply} (index multiply) of OGFs from Lemma \ref{th:OperationsOGF}. See Section 2.2.1 of \cite{arratia2019size} for discussions on the characteristic function and pgfs of size-biased rvs. From \eqref{eq:conditional-mean-generating-function}, we have
	\begin{align*}
		\mathcal{P}^{[1]}_{S}(t) = E\left[X_1 t^{S}\right] = \sum_{k = 0}^{\infty} t^{k}\sum_{k_1 = 0}^{k} k_1 f_{X_1, S_{-1}}(k_1, k-k_1) = \sum_{k = 0}^{\infty} t^{k}\sum_{k_1 = 0}^{k} \frac{f_{\widetilde{X}_1, S_{-1}}(k_1, k-k_1)}{E[X_1]},
	\end{align*}
	then $E[X_1]\mathcal{P}^{[1]}_{S}(t)$ is the pgf of $\widetilde{X}_1 + S_{-1}$, so \eqref{eq:allocation-size-biased} follows immediately. 

Future research could involve developing methods to quantify or correct aliasing errors for heavy-tailed distributions. In Section \ref{ss:application3}, we use a very large truncation point ($k_{max} = 2^{20}$). As computer processors continue to perform faster computations, it is convenient to increase the truncation point; however, it may also be convenient to provide methods that reduce this error source for efficiency's sake. The authors of \cite{grubel1999computation} quantify the aliasing error related to using the FFT algorithm to compute the pmf of compound distributions and propose a tilting procedure to reduce this error. Developing a similar theory for the OGFs of unconditional expected allocations and cumulative unconditional expected allocations will increase these methods' efficiency. 

Another research topic involves the allocation of tail variance. In \cite{furman2006tail}, the authors introduce the tail variance, defined by 
$$TV_{\kappa}(X) = Var(X \vert X > F_{X}^{-1}(\kappa)),$$
with $\kappa \in (0, 1)$, and propose allocations via the tail covariance allocation rule, 
$$TCov_{\kappa}(X_1 \vert S) = Cov(X_1, S \vert S > F_{S}^{-1}(\kappa)) = \sum_{j = 1}^{n}Cov(X_1, X_j \vert S > F_{S}^{-1}(\kappa)).$$
One can obtain efficient algorithms to compute the desired expectations once again. We have
$$E\left[X_1 X_j t^{S}\right] = \left. \left\{t_1 t_j \frac{\partial^2}{\partial t_1 \partial t_j} \mathcal{P}_{X_1, \dots, X_n}(t_1, \dots, t_n) \right\}\right\vert_{t_1 = \dots =t_n = t}$$
for $j \in \{1, \dots, n\}\setminus \{1\}$. The OGF for unconditional expected allocations for the second factorial moment is 
$$E\left[X_1(X_1 - 1) t^S\right] = \left.\left\{t_1^2 \frac{\partial^2}{\partial t_1^2} \mathcal{P}_{X_1, S_{-1}}(t_1, t_2)\right\}\right\vert_{t_1 = t_2 = t},$$
one can generalize the latter formula to $k$th factorial moments by taking subsequent derivatives. It follows that 
$E\left[X_1 X_j \vert S > k\right]$ and $E\left[X_1^2 \vert S > k\right]$ can be computed with 
$$E\left[X_1 X_j \times 1_{\{S \leq k\}}\right] = [t^k] \left\{\frac{E\left[X_1 X_j t^{S}\right]}{1-t}\right\}$$
and
$$E\left[X_1^2 \times 1_{\{S \leq k\}}\right] = [t^k] \left\{\frac{E\left[X_1(X_1 - 1) t^S\right] + \mathcal{P}_{S}^{[1]}(t)}{1-t}\right\}.$$

Finally, one can consider the implications of this method in the continuous case. Letting $\mathcal{L}_{X_1, \dots, X_n}$ denote the multivariate Laplace-Stieltjes transform of the vector $(X_1, \dots, X_n)$, one can show that $$\left.-\frac{\partial}{\partial t_1} \mathcal{L}_{X_1, \dots, X_n}(t_1, \dots, t_n)\right\vert_{t_1 = \dots = t_n = t},$$ for $t \geq 0$, is the Laplace transform of $E\left[X_1 \times 1_{\{S = s\}}\right]$. One could use this formulation to obtain new closed-form expressions for unconditional expected allocations, compute unconditional expected allocations through numerical inversion of Laplace transforms, or develop asymptotic properties of unconditional expected allocations. We note that the Laplace transform of size-biased rvs in the context of continuous rvs is explored in, for instance, \cite{furman2020lognormal}.

\section{Acknowledgements}

This work was partially supported by the Natural Sciences and Engineering Research Council of Canada (Cossette: 04273; Marceau: 05605). We thank Christian Y. Robert for fruitful discussions. We also thank the Editor and the reviewers who have provided helpful comments to improve the quality of this paper.

\appendix

\section{R code for the numerical applications}\label{app:code}

\subsection{Small portfolio of independent compound Poisson rvs}\label{app:small}
\begin{lstlisting}[lastline=32]
n_participants <- 4
kmax <- 2^6
cmax <- 4
lam <- list(0.08, 0.08, 0.1, 0.1)
fc <- list(c(0, 0.1, 0.2, 0.4, 0.3, rep(0, kmax - cmax - 1)),
           c(0, 0.15, 0.25, 0.3, 0.3, rep(0, kmax - cmax - 1)),
           c(0, 0.1, 0.2, 0.3, 0.4, rep(0, kmax - cmax - 1)),
           c(0, 0.15, 0.25, 0.3, 0.3, rep(0, kmax - cmax - 1)))

dft_fx <- list()
phic <- list()
mu <- list()
conditional_mean <- list()

for(i in 1:n_participants) {
  dft_fx[[i]] <- exp(lam[[i]] * (fft(fc[[i]]) - 1))
  phic[[i]] <- fft(c(1:cmax * fc[[i]][2:(cmax + 1)], rep(0, kmax - cmax)))
}

dft_fs <- Reduce("*", dft_fx)
fs <- Re(fft(dft_fs, inverse = TRUE))/kmax
e1 <- exp(-2i*pi*(0:(kmax-1))/kmax)
  
for(i in 1:n_participants) {
  dft_mu <- e1 * phic[[i]] * lam[[i]] * dft_fs
  mu[[i]] <- Re(fft(dft_mu, inverse = TRUE))/kmax
  conditional_mean[[i]] <- mu[[i]]/fs
}

sapply(conditional_mean, "[[", 2) # Validation
conditional_mean_total <- Reduce("+", conditional_mean)
conditional_mean_total[1 + 1:10] # Validation
\end{lstlisting}

\subsection{Large portfolio of independent compound Poisson rvs}\label{app:large}
\begin{lstlisting}[lastline=39]
set.seed(10112021)
n_participants <- 10000
kmax <- 2^13
lam <- list()
fc <- list()
mu <- list()

lambdas <- rexp(n_participants, 10)
rs <- sample(1:6, n_participants, replace = TRUE)
qs <- runif(n_participants, 0.4, 0.5)

# Assign parameters
for(i in 1:n_participants) {
  lam[[i]] <- lambdas[i]
  fci <- dnbinom(0:(kmax-2), rs[i], qs[i])
  fc[[i]] <- c(fci, 1 - sum(fci))
}

dft_fx <- list()
phic <- list()
cm <- list()

for(i in 1:n_participants) {
  dft_fx[[i]] <- exp(lam[[i]] * (fft(fc[[i]]) - 1))
  phic[[i]] <- fft(c(1:(kmax-1) * fc[[i]][-1], 0))
}

dft_fs <- Reduce("*", dft_fx)
fs <- Re(fft(dft_fs, inverse = TRUE))/kmax
e1 <- exp(-2i*pi*(0:(kmax-1))/kmax)

for(i in 1:n_participants) {
  dft_mu <- e1 * phic[[i]] * lam[[i]] * dft_fs
  mu[[i]] <- Re(fft(dft_mu, inverse = TRUE))/kmax
  cm[[i]] <- mu[[i]]/fs
}

cm_tot <- Reduce("+", cm)
cm_tot[1 + seq(4000, 5000, 100)] # Validation
\end{lstlisting}

\subsection{Portfolio of heavy-tailed risks}\label{app:heavy}
\begin{lstlisting}[lastline=29]
library(actuar)
n <- 3

xmax <- 2^15
kmax <- 2^20
alphas <- seq(1.3, 1.9, 0.3)
lambdas <- 10 * (alphas - 1)

phis <- rep(1, kmax)
cm3 <- list()

for(i in 1:n) {
  fx <- discretize(ppareto(x, alphas[i], lambdas[i]), 0, xmax - 1, 
                   method = "unbiased", lev = levpareto(x, alphas[i], lambdas[i]))
  phix <- fft(c(fx, rep(0, kmax - xmax)))
  phis <- phis * phix
}

fs3 <- Re(fft(phis, inverse = TRUE))/kmax
good_values <- (fs3 >= 0)

for(i in 1:n) {
  fx <- discretize(ppareto(x, alphas[i], lambdas[i]), 0, xmax - 1, 
                   method = "unbiased", lev = levpareto(x, alphas[i], lambdas[i]))
  phix <- fft(c(fx, rep(0, kmax - xmax)))
  phi_deriv_x1 <- fft(c((1:(xmax - 1)) * fx[-1], rep(0, kmax - xmax + 1)))
  agf <- phis / phix * phi_deriv_x1 * exp(-2i*pi*(0:(kmax-1))/kmax)
  cm3[[i]] <- (Re(fft(agf, inverse = TRUE))/kmax / fs3)[1:xmax]
}
\end{lstlisting}

\subsection{Archimedean copula example}\label{app:arch}
\begin{lstlisting}
set.seed(20220314)
n <- 6
# bi <- sample(1:10, n, replace = TRUE)
bi <- c(1, 3, 10, 4, 5, 10)
qi <- c(0.1, 0.15, 0.2, 0.25, 0.3)

kmax <- sum(bi) + 1
fft1 <- exp(-2i * pi * (0:(kmax - 1))/kmax)

alph <- 0
eps_theta <- 1e-10
theta_max <- max(2, floor(log(eps_theta)/log(alph)) + 1)
f_theta <- alph**(1:theta_max - 1) * (1 - alph)

LST_inv_geom <- function(u) log((1 - alph)/u + alph)
fft1 <- exp(-2i * pi * (0:(kmax - 1))/kmax)

qi <- runif(n)
ri <- exp(-LST_inv_geom(qi))

fgp_S <- function(s) {
	marginals <- apply(sapply(1:n, function(k) 1 - ri[k]^(1:theta_max) + ri[k]^(1:theta_max) * s^bi[k]), 1, prod)
	sum(f_theta * marginals)
}

fgp_S <- Vectorize(fgp_S)
phis <- fgp_S(fft1)

fs <- (Re(fft(phis, inverse = TRUE))/kmax)

pgf_alloc_i <- function(s, i) {
	marginals <- bi[i] * ri[i]^(1:theta_max) * s^bi[i] * apply(sapply((1:n)[-i], function(k) 1 - ri[k]^(1:theta_max) + ri[k]^(1:theta_max) * s^bi[k]), 1, prod)
	sum(f_theta * marginals)
}

pgf_alloc_i <- Vectorize(pgf_alloc_i)
phi_alloc_1 <- pgf_alloc_i(fft1, 1)
conditional_mean_1 <- (Re(fft(phi_alloc_1, inverse = TRUE))/kmax/fs)
round(conditional_mean_1, 3)
plot(conditional_mean_1, type = 's')
\end{lstlisting}


\begin{thebibliography}{}

	\bibitem[Arratia et~al., 2019]{arratia2019size}
	Arratia, R., Goldstein, L., and Kochman, F. (2019).
	\newblock Size bias for one and all.
	\newblock {\em Probability Surveys}, 16:1--61.
	
	\bibitem[Axelrod and Kimmel, 2015]{axelrod2015branching}
	Axelrod, D. and Kimmel, M. (2015).
	\newblock {\em Branching Processes in Biology}.
	\newblock Springer.
	
	\bibitem[Barg{\`e}s et~al., 2009]{barges2009tvarbased}
	Barg{\`e}s, M., Cossette, H., and Marceau, E. (2009).
	\newblock {{TVaR}}-based capital allocation with copulas.
	\newblock {\em Insurance: Mathematics and Economics}, 45(3):348--361.
	
	\bibitem[Cooley and Tukey, 1965]{cooley1965algorithm}
	Cooley, J.~W. and Tukey, J.~W. (1965).
	\newblock An algorithm for the machine calculation of complex {{Fourier}}
	  series.
	\newblock {\em Mathematics of Computation}, 19(90):297--301.
	
	\bibitem[Cormen et~al., 2009]{cormen2009introductiona}
	Cormen, T.~H., Leiserson, C.~E., Rivest, R.~L., and Stein, C. (2009).
	\newblock {\em Introduction to {{Algorithms}}}.
	\newblock {MIT Press}.
	
	\bibitem[Cossette et~al., 2018]{cossette2018dependent}
	Cossette, H., Marceau, E., Mtalai, I., and Veilleux, D. (2018).
	\newblock Dependent risk models with {{Archimedean}} copulas: {{A}}
	  computational strategy based on common mixtures and applications.
	\newblock {\em Insurance: Mathematics and Economics}, 78:53--71.
	
	\bibitem[Denuit, 2019]{denuit2019sizebiased}
	Denuit, M. (2019).
	\newblock Size-biased transform and conditional mean risk sharing, with
	  application to {{P2P}} insurance and tontines.
	\newblock {\em ASTIN Bulletin: The Journal of the IAA}, 49(3):591--617.
	
	\bibitem[Denuit, 2020]{denuit2020sizebiased}
	Denuit, M. (2020).
	\newblock Size-{{Biased Risk Measures}} of {{Compound Sums}}.
	\newblock {\em North American Actuarial Journal}, 24(4):512--532.
	
	\bibitem[Denuit and Dhaene, 2012]{denuit2012convex}
	Denuit, M. and Dhaene, J. (2012).
	\newblock Convex order and comonotonic conditional mean risk sharing.
	\newblock {\em Insurance: Mathematics and Economics}, 51(2):265--270.
	
	\bibitem[Denuit et~al., 2006]{denuit2006actuarial}
	Denuit, M., Dhaene, J., Goovaerts, M., and Kaas, R. (2006).
	\newblock {\em Actuarial {{Theory}} for {{Dependent Risks Measures}},
	  {{Orders}} and {{Models}}}.
	\newblock Wiley.
	
	\bibitem[Denuit et~al., 2022]{denuit2021risksharing}
	Denuit, M., Dhaene, J., and Robert, C.~Y. (2022).
	\newblock Risk-sharing rules and their properties, with applications to
	  peer-to-peer insurance.
	\newblock {\em Journal of Risk and Insurance}, 89(3):615--667.
	
	\bibitem[Denuit et~al., 2021]{denuit2021mortality}
	Denuit, M., Hieber, P., and Robert, C.~Y. (2021).
	\newblock Mortality credits within large survivor funds.
	\newblock {\em ASTIN Bulletin: The Journal of the IAA}, 52:1--22.
	
	\bibitem[Denuit and Robert, 2020]{denuit2020largeloss}
	Denuit, M. and Robert, C.~Y. (2020).
	\newblock Large-loss behavior of conditional mean risk sharing.
	\newblock {\em ASTIN Bulletin: The Journal of the IAA}, 50(3):1093--1122.
	
	\bibitem[Denuit and Robert, 2021a]{denuit2021conditionala}
	Denuit, M. and Robert, C.~Y. (2021a).
	\newblock Conditional tail expectation decomposition and conditional mean risk
	  sharing for dependent and conditionally independent losses.
	\newblock {\em Methodology and Computing in Applied Probability}, 24.
	
	\bibitem[Denuit and Robert, 2021b]{denuit2021risk}
	Denuit, M. and Robert, C.~Y. (2021b).
	\newblock From risk sharing to pure premium for a large number of heterogeneous
	  losses.
	\newblock {\em Insurance: Mathematics and Economics}, 96:116--126.
	
	\bibitem[Dickson, 2017]{dickson2017insurance}
	Dickson, D. C.~M. (2017).
	\newblock {\em Insurance Risk and Ruin}.
	\newblock International Series on Actuarial Science. {Cambridge University
	  Press}.
	
	\bibitem[Embrechts and Frei, 2009]{embrechts2009panjer}
	Embrechts, P. and Frei, M. (2009).
	\newblock Panjer recursion versus {{FFT}} for compound distributions.
	\newblock {\em Mathematical Methods of Operations Research}, 69(3):497--508.
	
	\bibitem[Embrechts et~al., 1993]{embrechts1993applications}
	Embrechts, P., Gr{\"u}bel, R., and Pitts, S.~M. (1993).
	\newblock Some applications of the fast {{Fourier}} transform algorithm in
	  insurance mathematics {{This}} paper is dedicated to {{Professor W}}. {{S}}.
	  {{Jewell}} on the occasion of his 60th birthday.
	\newblock {\em Statistica Neerlandica}, 47(1):59--75.
	
	\bibitem[Embrechts and Hofert, 2013]{embrechts2013note}
	Embrechts, P. and Hofert, M. (2013).
	\newblock A note on generalized inverses.
	\newblock {\em Mathematical Methods of Operations Research}, 77(3):423--432.
	
	\bibitem[Flajolet and Sedgewick, 2009]{flajolet2009analytic}
	Flajolet, P. and Sedgewick, R. (2009).
	\newblock {\em Analytic Combinatorics}.
	\newblock Cambridge University Press.
	
	\bibitem[Furman, 2007]{furman2007convolution}
	Furman, E. (2007).
	\newblock On the convolution of the negative binomial random variables.
	\newblock {\em Statistics \& Probability Letters}, 77(2):169--172.
	
	\bibitem[Furman et~al., 2020]{furman2020lognormal}
	Furman, E., Hackmann, D., and Kuznetsov, A. (2020).
	\newblock On log-normal convolutions: {{An}} analytical\textendash numerical
	  method with applications to economic capital determination.
	\newblock {\em Insurance: Mathematics and Economics}, 90:120--134.
	
	\bibitem[Furman and Landsman, 2005]{furman2005risk}
	Furman, E. and Landsman, Z. (2005).
	\newblock Risk capital decomposition for a multivariate dependent gamma
	  portfolio.
	\newblock {\em Insurance: Mathematics and Economics}, 37(3):635--649.
	
	\bibitem[Furman and Landsman, 2006]{furman2006tail}
	Furman, E. and Landsman, Z. (2006).
	\newblock Tail variance premium with application for elliptical portfolio of
	  risks.
	\newblock {\em ASTIN Bulletin: The Journal of the IAA}, 36(2):433--462.
	
	\bibitem[Furman and Landsman, 2008]{furman2008economic}
	Furman, E. and Landsman, Z. (2008).
	\newblock Economic capital allocations for non-negative portfolios of dependent
	  risks.
	\newblock {\em ASTIN Bulletin: The Journal of the IAA}, 38(2):601--619.
	
	\bibitem[Graham et~al., 1994]{graham1994concrete}
	Graham, R.~L., Knuth, D.~E., and Patashnik, O. (1994).
	\newblock {\em Concrete Mathematics: A Foundation for Computer Science}.
	\newblock {Addison-Wesley}, 2nd edition.
	
	\bibitem[Grimmett and Stirzaker, 2020]{grimmett2020probability}
	Grimmett, G. and Stirzaker, D. (2020).
	\newblock {\em Probability and {{Random Processes}}}.
	\newblock {Oxford University Press}.
	
	\bibitem[Grubel and Hermesmeier, 1999]{grubel1999computation}
	Grubel, R. and Hermesmeier, R. (1999).
	\newblock Computation of compound distributions i: Aliasing errors and
	  exponential tilting.
	\newblock {\em ASTIN Bulletin: The Journal of the IAA}, 29(2):197--214.
	
	\bibitem[Jiao et~al., 2022]{jiao2022axiomatic}
	Jiao, Z., Liu, Y., and Wang, R. (2022).
	\newblock An axiomatic theory for anonymized risk sharing.
	\newblock {\em arXiv preprint arXiv:2208.07533}.
	
	\bibitem[Johnson et~al., 1997]{johnson1997discrete}
	Johnson, N.~L., Kotz, S., and Balakrishnan, N. (1997).
	\newblock {\em Discrete {{Multivariate Distributions}}}.
	\newblock {John Wiley \& Sons}.
	
	\bibitem[Johnson et~al., 2005]{johnson2005univariate}
	Johnson, N.~L., Kotz, S., and Kemp, A.~W. (2005).
	\newblock {\em Univariate Discrete Distributions}.
	\newblock John Wiley \& Sons.
	
	\bibitem[Katz, 1965]{katz1965unified}
	Katz, L. (1965).
	\newblock Unified treatment of a broad class of discrete probability
	  distributions.
	\newblock {\em Classical and Contagious Discrete Distributions}, 1:175--182.
	
	\bibitem[Klugman et~al., 2018]{klugman2018loss}
	Klugman, S.~A., Panjer, H.~H., and Willmot, G.~E. (2018).
	\newblock {\em Loss Models: From Data to Decisions}.
	\newblock Wiley Series in Probability and Statistics. {Society of Actuaries,
	  Wiley}.
	
	\bibitem[Kocherlakota, 1992]{kocherlakota1992bivariate}
	Kocherlakota (1992).
	\newblock {\em Bivariate {{Discrete Distributions}}}.
	\newblock {CRC Press}.
	
	\bibitem[Kocherlakota, 1988]{kocherlakota1988compounded}
	Kocherlakota, S. (1988).
	\newblock On the compounded bivariate {{Poisson}} distribution: {{A}} unified
	  treatment.
	\newblock {\em Annals of the Institute of Statistical Mathematics},
	  40(1):61--76.
	
	\bibitem[Lindskog and McNeil, 2003]{lindskog2003common}
	Lindskog, F. and McNeil, A.~J. (2003).
	\newblock Common {{Poisson}} shock models: Applications to insurance and credit
	  risk modelling.
	\newblock {\em ASTIN Bulletin: The Journal of the IAA}, 33(2):209--238.
	
	\bibitem[Mahamunulu, 1967]{mahamunulu1967note}
	Mahamunulu, D.~M. (1967).
	\newblock A note on regression in the multivariate {{Poisson}} distribution.
	\newblock {\em Journal of the American Statistical Association},
	  62(317):251--258.
	
	\bibitem[Marceau, 2013]{marceau2013modelisation}
	Marceau, {\'E}. (2013).
	\newblock {\em {Mod\'elisation et Évaluation Quantitative des Risques en
	  Actuariat}}.
	\newblock {Springer}.
	
	\bibitem[Marshall and Olkin, 1988]{marshall1988families}
	Marshall, A.~W. and Olkin, I. (1988).
	\newblock Families of multivariate distributions.
	\newblock {\em Journal of the American Statistical Association},
	  83(403):834--841.
	
	\bibitem[Mathai and Moschopoulos, 1991]{mathai1991multivariate}
	Mathai, A.~M. and Moschopoulos, P.~G. (1991).
	\newblock On a multivariate gamma.
	\newblock {\em Journal of Multivariate Analysis}, 39(1):135--153.
	
	\bibitem[McNeil et~al., 2015]{mcneil2015quantitative}
	McNeil, A.~J., Frey, R., and Embrechts, P. (2015).
	\newblock {\em Quantitative Risk Management: Concepts, Techniques and Tools}.
	\newblock Princeton Series in Finance. {Princeton University Press}.
	
	\bibitem[Muller and Stoyan, 2002]{muller2002comparison}
	Muller, A. and Stoyan, D. (2002).
	\newblock {\em Comparison {M}ethods for {S}tochastic {M}odels and {R}isks}.
	\newblock Wiley.
	
	\bibitem[Panjer, 1981]{panjer1981recursivea}
	Panjer, H.~H. (1981).
	\newblock Recursive evaluation of a family of compound distributions.
	\newblock {\em ASTIN Bulletin: The Journal of the IAA}, 12(1):22--26.
	
	\bibitem[Panjer and Willmot, 1992]{panjer1992insurance}
	Panjer, H.~H. and Willmot, G.~E. (1992).
	\newblock {\em Insurance Risk Models}.
	\newblock Society of Actuaries.
	
	\bibitem[Ren and Zitikis, 2017]{ren2017cmph}
	Ren, J. and Zitikis, R. (2017).
	\newblock {{CMPH}}: A multivariate phase-type aggregate loss distribution.
	\newblock {\em Dependence Modeling}, 5(1):304--315.
	
	\bibitem[Sedgewick and Flajolet, 2013]{sedgewick2013introduction}
	Sedgewick, R. and Flajolet, P. (2013).
	\newblock {\em An Introduction to the Analysis of Algorithms}.
	\newblock {Addison-Wesley}.
	
	\bibitem[Subrahmaniam, 1966]{subrahmaniam1966test}
	Subrahmaniam, K. (1966).
	\newblock A test for "intrinsic correlation" in the theory of accident
	  proneness.
	\newblock {\em Journal of the Royal Statistical Society. Series B
	  (Methodological)}, 28(1):180--189.
	
	\bibitem[Tasche, 1999]{tasche1999risk}
	Tasche, D. (1999).
	\newblock Risk contributions and performance measurement.
	\newblock {\em Report of the Lehrstuhl f{\"u}r Mathematische Statistik, TU
	  M{\"u}nchen}.
	
	\bibitem[Teicher, 1954]{teicher1954multivariate}
	Teicher, H. (1954).
	\newblock On the multivariate {{Poisson}} distribution.
	\newblock {\em Scandinavian Actuarial Journal}, 1954(1):1--9.
	
	\bibitem[Wang, 1996]{wang1996premium}
	Wang, S. (1996).
	\newblock Premium calculation by transforming the layer premium density.
	\newblock {\em ASTIN Bulletin: The Journal of the IAA}, 26(1):71--92.
	
	\bibitem[Wang, 1998]{wang1998aggregation}
	Wang, S.~S. (1998).
	\newblock Aggregation of correlated risk portfolios: Models and algorithms.
	\newblock {\em Proceedings of the Casualty Actuarial Society}, pages 848--939.
	
	\bibitem[Wilf, 2006]{wilf2006generatingfunctionology}
	Wilf, H.~S. (2006).
	\newblock {\em Generatingfunctionology}.
	\newblock CRC Press.
	
	\bibitem[Winkelmann, 2008]{winkelmann2008econometric}
	Winkelmann, R. (2008).
	\newblock {\em Econometric Analysis of Count Data}.
	\newblock Springer.
	
	\bibitem[Zhang et~al., 2018]{zhang2018generalized}
	Zhang, M., Hong, Y., and Balakrishnan, N. (2018).
	\newblock The generalized {{Poisson-binomial}} distribution and the computation
	  of its distribution function.
	\newblock {\em Journal of Statistical Computation and Simulation},
	  88(8):1515--1527.
	
	\end{thebibliography}
\end{document}